\documentclass[runningheads]{llncs}
\usepackage{graphicx}
\usepackage{pbox}
\usepackage{amsmath}
\usepackage{rotating}
\usepackage{amssymb}
\usepackage{geometry}
\usepackage{multirow}
\usepackage{algorithm}
\usepackage{algpseudocode}
\usepackage{mathtools}
\usepackage{comment}
\usepackage{bbm}
\usepackage[misc,geometry]{ifsym}
\usepackage[british]{babel}

\newcommand{\BigO}{\mathcal{O}}
\newcommand{\BigOs}{\mathcal{O}^*}
\newcommand{\DFT}{\textrm{DFT}}
\newcommand{\Field}{\mathbb{F}}
\newcommand{\InvDFT}{\DFT^{-1}}
\newcommand{\Int}{\mathbb{Z}}
\newcommand{\Nat}{\mathbb{N}}
\newcommand{\NP}{\mathcal{NP}}
\newcommand{\tens}{\vec}

\begin{document}
\title{Fast Algorithms for Join Operations on Tree Decompositions}
%
%
\author{Johan M. M. van Rooij (\Letter)}
\authorrunning{J.M.M. van Rooij}
%
\institute{Department of Information and Computing Sciences, Utrecht University\\PO Box 80.089, 3508 TB Utrecht, The Netherlands\\\email{J.M.M.vanRooij@uu.nl}}
\maketitle

\begin{abstract}
	Treewidth is a measure of how tree-like a graph is.
	It has many important algorithmic applications because many NP-hard problems on general graphs become tractable when restricted to graphs of bounded treewidth.
	Algorithms for problems on graphs of bounded treewidth mostly are dynamic programming algorithms using the structure of a tree decomposition of the graph.
	The bottleneck in the worst-case run time of these algorithms often is the computations for the so called join nodes in the associated nice tree decomposition.
	
	In this paper, we review two different approaches that have appeared in the literature about computations for the join nodes: one using fast zeta and M\"obius transforms and one using fast Fourier transforms.
	We combine these approaches to obtain new, faster algorithms for a broad class of vertex subset problems known as the $[\sigma,\rho]$-domination problems.	
	Our main result is that we show how to solve $[\sigma,\rho]$-domination problems in $\BigO(s^{t+2} t n^2 (t\log(s)+\log(n)))$ arithmetic operations.
	Here, $t$ is the treewidth, $s$ is the (fixed) number of states required to represent partial solutions of the specific $[\sigma,\rho]$-domination problem, and $n$ is the number of vertices in the graph.
	This reduces the polynomial factors involved compared to the previously best time bound (van Rooij, Bodlaender, Rossmanith, ESA 2009) of $\BigO( s^{t+2} (st)^{2(s-2)} n^3 )$ arithmetic operations.
	In particular, this removes the dependence of the degree of the polynomial on the fixed number of states~$s$.
	
\end{abstract}

\section{Introduction}
Treewidth is an important concept in the theory of graph algorithms that measures how tree-like a graph is.
While many problems that are $\NP$-hard on general graphs become efficiently solvable when restricted to trees, this often extends to these problems being polynomial or even linear-time solvable when restricted to graphs that have bounded treewidth.
In general this is done in two steps:
\begin{enumerate}
	\item Find a tree decomposition of the input graph of small treewidth.
	\item Solve the problem by dynamic programming on this tree decomposition.
\end{enumerate}
In this paper, we focus on the second of these two steps and show how to improve the running times of algorithms on tree decompositions using algebraic transforms.
We apply these to the general case of the so called $[\sigma,\rho]$-domination problems.
This includes many well-known vertex subset problems such as {\sc Independent Set}, {\sc Dominating Set} and {\sc Total Dominating Set}, but also problems such as {\sc Induced Bounded Degree Subgraph} and {\sc Induced $p$-Regular Subgraph}.

If we assume that a graph~$G$ is given with a tree decomposition~$T$ of~$G$ of width~$t$, then the running time of an algorithm on tree decompositions is typically polynomial in the size of graph~$G$, but exponential in the treewidth~$t$.
Early examples of such algorithms include algorithms on vertex partitioning problems (including the $[\rho,\sigma]$-domination problems)~\cite{TelleP97}, edge colouring problems such as {\sc Chromatic Index}~\cite{Bodlaender90}, or other problems such as {\sc Steiner Tree}~\cite{KorachS90}.
Often the worst-case running time of these algorithms involve large factors that depend on the treewidth~$t$.
This lead researchers to look for algorithms where these factors grow as slow as possible as a function of~$t$.
For several {\sc Dominating Set}-like problems such as {\sc Independent Dominating Set}, {\sc Total Dominating Set}, {\sc Perfect Dominating Set} and {\sc Perfect Code}, Alber et al.~\cite{AlberBFKN02} give improved algorithms with special attention to the exponential dependence on the treewidth~$t$: for example, they showed how to solve {\sc Dominating Set} in $\BigOs(4^t)$ time.
This was improved by Van Rooij et al.~\cite{vanRooijBR09} who first showed how to solve {\sc Dominating Set} in $\BigOs(3^t)$ time by giving an $\BigO(3^tt^2n)$-time algorithm.
Van Rooij et al.~also generalised this result solving the $[\rho,\sigma]$-domination problems in  $\BigO( s^{t+2} (st)^{2(s-2)} n^3 )$ time.
The result for {\sc Dominating Set} seems to be optimal in some sense, as Lokshtanov et al.~\cite{LokshtanovMS10} showed that any $\BigOs((3-\epsilon)^t)$-time algorithm would violate the \emph{Strong Exponential-Time Hypothesis}; we expect the same for the other $[\rho,\sigma]$-domination problems.

Since then, several results have appeared improving running times of dynamic programming algorithms on tree decompositions.
For example, the algorithm by Van Rooij et al.~\cite{vanRooijBR09} has been generalised to {\sc Distance-$r$ Dominating Set}~\cite{BorradailH17} and {\sc Distance-$r$ Independent Set}~\cite{KatsikarelisLP18}.
The most notable new results are the \emph{Cut and Count} technique~\cite{CyganNPPRW11} giving randomised $\BigOs(c^t)$-time algorithms for many graph connectivity problems, mostly supported by matching lower bounds based on the Strong Exponential-Time Hypotheses, the \emph{rank-based approach}~\cite{BodlaenderCKN15,CyganKN13} and the \emph{determinant-based approach}~\cite{BodlaenderCKN15,Wlodarczyk19} that derandomise these results at the cost of a greater base of the exponent~$c$.

For many of these algorithms, the computations in the so called join nodes of a nice tree decomposition are the bottleneck of the worst-case run time.
To speed up these computations, several approaches have been used, often based on algebraic transforms.
One such method is using fast zeta and M\"obius transforms in a way that is similar to the well-known fast subset convolution algorithm by Bj\"orklund et al~\cite{BjorklundHKK07}.
This method was first used in the context of tree decompositions by Van Rooij et al.~\cite{vanRooijBR09} who also generalised the approach to work for the $[\sigma,\rho]$-domination problems.
At the same time, Cygan and Pilipczuk, showed that the fast subset convolution result could also be based on Fourier transforms~\cite{CyganP10}; they also generalised it in a different way.
A variant to this approach that we follow in this paper, directly applied to tree decompositions, can be found in the appendix of~\cite{CyganNPPRW11a}.
We will discuss both these approaches in more detail in this paper.
Finally, faster joins are also obtained based on Clifford algebras~\cite{Wlodarczyk19}, but these are beyond the scope of this paper.

\subsection{Goal of this paper}
The goal of this paper is twofold.
Firstly, we want to present a faster algorithm for the $[\sigma,\rho]$-domination problems.
This algorithm uses $\BigO(s^{t+2} t n^2 (t\log(s)+\log(n)))$ arithmetic operations: this improves the polynomial factors compared to our earlier result~\cite{vanRooijBR09} and removes the dependency of the degree of the polynomial on $s$, where $s$ is the (fixed) number of states used.
Secondly, we want to give a comprehensible overview of how Fourier and M\"obius transforms can be used to obtain faster algorithms on tree decompositions.

We choose to take an algebraic perspective that allows for easier generalisation and easier combination of Fourier and M\"obius transform-based approaches than that in~\cite{vanRooij11,vanRooijBR09}.
However, we consider the approach in~\cite{vanRooij11,vanRooijBR09} to be more intuitive: it relies only on counting arguments (this is especially true for the first algorithm for {\sc Dominating Set} in~\cite{vanRooij11} that does not explicitly use any algebraic transform).
In our overview, we will not give details on our earlier generalised convolution approach from~\cite{vanRooijBR09}: after the initial examples, we directly go to the new and improved algorithm.

\section{Preliminaries} \label{sec:prelim}

\subsection{Graphs and Tree Decompositions} \label{sec:treewidth}
Let $G=(V,E)$ be an $n$-vertex graph with $m$ edges.
A \emph{terminal graph}\footnote{This is also known as a $k$-boundary graph} $G_X = (V,E,X)$ is a graph $G=(V,E)$ with an ordered sequence of distinct vertices that we call its terminals: $X=\{x_1,x_2,\ldots,x_k\}$ with each $x_j \in V$.
Two terminal graphs $G_X = (V_1,E_1,X_1)$ and $H_X=(V_2,E_2,X_2)$ with the same number of terminals~$k$, but disjoint vertex and edge sets, can be \emph{glued} together to form the terminal graph $G_X \oplus H_X$ by identifying each terminal $x_i$ from $X_1$ with $x_i$ from $X_2$, for all $1 \leq i \leq k$.
That is, if $X = X_1 = X_2$ through identification, then $G_X \oplus H_X = ( V_1 \cup V_2, E_1 \cup E_2, X)$.
A \emph{completion} of a terminal graph~$G_X$ is a non-terminal graph~$G$ that can be obtained from~$G_X$ by gluing a terminal graph~$H_X$ on $G_X$ and then ignoring which vertices are terminals in the result.

The treewidth of a (non-terminal) graph is a measure of how-tree like the graph is.
From an algorithmic viewpoint this is a very useful concept because, where many $\mathcal{NP}$-hard problems on general graphs are linear time solvable on trees by dynamic programming, often similar style dynamic programming algorithms exist for graphs whose treewidth is bounded by a constant.
We outline the basics on treewidth and specifically on dynamic programming on tree decompositions below.
More information can, amongst other places, be found in work by Bodlaender~\cite{Bodlaender88,Bodlaender93,Bodlaender97,Bodlaender98,BodlaenderK08}.
\begin{definition}[tree decomposition and treewidth] \label{def:tw}
	A \emph{tree decomposition} of an undirected graph~$G=(V,E)$ is a tree~$T$ in which each node $i \in T$ has an associated set of vertices $X_i \subseteq V$ (called a \emph{bag}), with $\bigcup_{i \in T} X_i = V$, such that the following properties hold:
	\begin{itemize}
		\item for every edge $\{u,v\} \in E$, there exist a bag $X_i$ such that $\{u,v\} \subseteq X_i$;
		\item for every vertex $v$~in $G$, the bags containing $v$ form a connected subtree: i.e., if $v \in X_i$ and $v \in X_j$, then $v \in X_k$ for all nodes $k$ on the path from $i$ to $j$ in $T$.
	\end{itemize}
	The \emph{width} of a tree decomposition~$T$ is defined as $\max_{i \in T}\{ |X_i| \} -1$: the size of the largest bag minus one.
	The \emph{treewidth} of a graph~$G$ is the minimum width over all tree decomposition of~$G$.
\end{definition}
For a tree decomposition~$T$ with assigned root node $r \in T$, we define the terminal graph $G_i = (V_i,E_i,X_i)$ for each node $i \in T$: let $V_i$ be the union of $X_i$ with all bags $X_j$ where~$j$ is a descendant of $i$ in $T$, and let $E_i \subseteq E$ be the set of edges with at least one endpoint in $V_i \setminus X_i$ (and as a result of Definition~\ref{def:tw} with both endpoints in $V_i$).
Now, $G_i$ contains all edges between vertices in $V_i \setminus X_i$, and all edges between $V_i \setminus X_i$ and $X_i$, but no edges between two vertices in $X_i$.\footnote{Often $G_i$ is defined \emph{including} all edges between vertices in $X_i$.
We choose the alternative definition as it makes formulating the join algorithms in Section~\ref{sec:joins} easier: no bookkeeping of number of neighbours between vertices in $X_i$ needs to be done, as they only become neighbours higher up in the tree.}
Observe that, $G$ is the completion of $G_i$ formed through $G_i \oplus ( (V \setminus V_i) \cup X_i, E \setminus E_i, X_i)$, and $X_i$ can be seen as the \emph{separator} separating $V_i \setminus X_i$ from $V \setminus V_i$ in $G$ (where either side of the separator can be empty). 

We now describe dynamic programming on a tree decomposition~$T$.
Given a graph problem that we are trying to solve~$\mathcal{P}$, define a \emph{partial solution} of~$\mathcal{P}$ on $G_i$ to be the \emph{restriction} to the subgraph~$G_i$ of a solution of $\mathcal{P}$ on a completion of $G_i$ (any completion of $G_i$, not only $G$ itself).
We say that the partial solution~$S'$ on $G_i$ can be \emph{extended} to a full solution~$S$ on a completion of $G_i$, where $S \setminus S'$ is the \emph{extension} of $S'$.
As an example, consider the {\sc Minimum Dominating Set} problem: a solution for this problem is a vertex subset~$D$ in $G$ such that for all $v \in V$ there is a $d \in D$ with $v \in N[d]$.
A partial solution is a subset $D \subseteq V_i$ such that for all vertices in $v \in V_i \setminus X_i$  there is a $d \in D$ with $v \in N[d]$: for vertices in $X_i$ there does not need to be a dominating neighbour in $d \in D$ as $d$ can also be in an extension of $D$.
A dynamic programming algorithm on a tree decomposition computes, for each node $i \in T$ in a bottom-up fashion, a \emph{memoisation table}~$A_i$ containing all \emph{relevant} (described in the next paragraph) partial solutions on $G_i$ obtaining a solution to $\mathcal{P}$ in the root of $T$.

To restrict the number of partial (relevant) solutions stored, an equivalence relation is defined on them: two partial solutions $S_1'$ and $S_2'$ on $G_i$ are \emph{equivalent} with respect to~$\mathcal{P}$ if any extension of $S_1$ also is an extension of $S_2'$ and vice versa.
When given two equivalent partial solutions $S_1'$ and $S_2'$ for an optimisation problem (minimisation or maximisation), we say that $S_1'$ \emph{dominates} $S_2'$ if for any extension $S_E$ of $S_1'$ and $S_2'$, the solution value of $S_1' \cup S_E$ is equal or better than the solution value of $S_2' \cup S_E$.
Clearly, a dynamic programming algorithm on a tree decomposition needs to store only one partial solution per equivalence class, and if we consider an optimisation problem it can store a partial solution that dominates all other partial solutions within its equivalence class.

Mostly, it is convenient to formulate a dynamic programming algorithm on a special kind of tree decomposition called a \emph{nice tree decomposition}~\cite{Kloks94}.\footnote{Different version of the original definition~\cite{Kloks94} exists in literature (e.g, \cite{CyganNPPRW11,vanRooij11}): the restrictions on the vertices in a bag of a leaf node and the root node often vary, and sometimes an additional type of node called an \emph{edge introduce node} is used.}
\begin{definition}[nice tree decomposition] \label{def:nicetd}
	A \emph{nice tree decomposition} is a tree decomposition~$T$ with assigned root node $r \in T$ with $X_r = \emptyset$, in which each node is of one of the following types:
	\begin{itemize}
		\item \emph{Leaf node}: a leaf $i$ of $T$ with $X_i = \emptyset$.
		\item \emph{Introduce node}: an internal node $i$ of $T$ with one child node $j$ and $X_i = X_j \cup \{v\}$ for some $v \in V \setminus V_j$.
		\item \emph{Forget node}: an internal node $i$ of $T$ with one child node $j$ and $X_i = X_j \setminus \{v\}$ for some $v \in X_j$.
		\item \emph{Join node}: an internal node $i$ of $T$ with two child nodes $l$ and $r$ with $X_i = X_l = X_r$.
	\end{itemize}
\end{definition}
Given a tree decomposition consisting of $\BigO(n)$ nodes, a nice tree decomposition of $\BigO(n)$ nodes of the same width can be found in $\BigO(n)$ time~\cite{Kloks94}.
Consequently, a dynamic programming algorithm on a nice tree decomposition can be used on general tree decompositions by applying this transformation.
After computing $A_i$ for all nodes $i \in T$, the solution to~$\mathcal{P}$ can be found as the unique value in $A_r$, where~$r$ is the root of $T$: here $G_i = G$ and there is only a single equivalence class as $X_i =\emptyset$.

This paper focuses on computing $A_i$ for a \emph{join node}~$i$ of a nice tree decomposition.
This node is the most interesting as often it dominates the running time of the entire dynamic programming algorithm.
For an example, consider~\cite{AlberBFKN02} where an $\BigOs(4^t)$ algorithm for {\sc Minimum Dominating Set} for graphs with a tree decomposition of width~$t$ is given, while all computations except the computation for the join nodes can be performed in $\BigOs(3^t)$ time.

\subsection{Dynamic Programming for $[\sigma,\rho]$-Domination Problems} \label{sec:sigmarho}
The $[\sigma,\rho]$-domination problems are a class of vertex-subset problems introduced by Telle~\cite{Telle94,Telle94a,TelleP97} that generalise many well-known graph problems such as {\sc Maximum Independent Set}, {\sc Minimum Dominating Set}, and {\sc Induced Bounded Degree Subgraph}.
See Table~\ref{tab:sigmarho} for an overview.
\begin{definition}[{$[\sigma,\rho]$}-dominating set] \label{def:sigmarho}
	Let $\sigma, \rho \subseteq \Nat$, a $[\sigma,\rho]$-dominating set in a graph $G=(V,E)$ is a subset $D \subseteq V$ such that:
	\begin{itemize}
		\item for every $v \in D$: $|N(v) \cap D| \in \sigma$;
		\item for every $v \in V \setminus D$: $|N(v) \cap D| \in \rho$.
	\end{itemize}
\end{definition}
We consider only $\sigma, \rho \subseteq \Nat$ that both are either finite or cofinite.

\begin{table}[tbp]
	\begin{center}
		\begin{tabular}{|l|l||l|}
			\hline
			$\sigma$ & $\rho$ & Standard description \\
			\hline
			$\{0\}$ & $\{0,1,\ldots\}$ & Independent Set/Stable Set\\
			$\{0,1,\ldots\}$ & $\{1,2,\ldots\}$ & Dominating Set\\
			$\{0\}$ & $\{0,1\}$ & Strong Stable Set/2-Packing/Distance-2 Independent Set\\
			$\{0\}$ & $\{1\}$ & Perfect Code/Efficient Dominating Set\\
			$\{0\}$ & $\{1,2,\ldots\}$ & Independent Dominating Set\\
			$\{0,1,\ldots\}$ & $\{1\}$ & Perfect Dominating Set\\
			$\{1,2,\ldots\}$ & $\{1,2,\ldots\}$ &Total Dominating Set\\
			$\{1\}$ & $\{1\}$ & Total Perfect Dominating Set\\
			$\{0,1,\ldots\}$ & $\{0,1\}$ & Nearly Perfect Set\\
			$\{0,1\}$ & $\{0,1\}$ & Total Nearly Perfect Set\\
			$\{0,1\}$ & $\{1\}$ & Weakly Perfect Dominating Set\\
			$\{0,1,\ldots,p\}$ & $\{0,1,\ldots\}$ & Induced Bounded Degree Subgraph\\
			$\{0,1,\ldots\}$ & $\{p,p+1,\ldots\}$ & $p$-Dominating Set\\
			$\{p\}$ & $\{0,1,\ldots\}$ & Induced $p$-Regular Subgraph\\
			\hline
		\end{tabular}
	\end{center}
	\caption{Examples of $[\sigma,\rho]$-domination problems (taken from \cite{Telle94,Telle94a,TelleP97}).}
	\label{tab:sigmarho}
\end{table}

For given $\sigma, \rho \subseteq \Nat$ and the corresponding definition of a $[\sigma,\rho]$-dominating set, one can define several different problem variants.
\begin{itemize}
	\item \emph{Existence} problem: given a graph $G$, does $G$ have a $[\sigma, \rho]$-dominating set?
	\item \emph{Optimisation} problem (minimisation or maximisation): given a graph $G$, what is the smallest $[\sigma, \rho]$-dominating set in $G$, or what is the largest $[\sigma, \rho]$-dominating set in $G$?
	\item \emph{Counting} problem: given a graph $G$, how many $[\sigma, \rho]$-dominating sets exist in $G$?
	\item \emph{Counting optimisation} problem (minimisation or maximisation): given a graph $G$, how many $[\sigma, \rho]$-dominating sets in $G$ exist of minimum/maximum size?
\end{itemize}
Many well-known NP-hard vertex subset problems in graphs correspond to the existence or optimisation variant of a $[\sigma,\rho]$-domination problem, as can be seen from Table~\ref{tab:sigmarho}.

When solving a $[\sigma,\rho]$-domination problem by dynamic programming on a tree decomposition, the equivalence classes for partial solutions stored in the memoisation table~$A_i$ (as defined in Section~\ref{sec:treewidth}) can be uniquely identified by the following:
\begin{itemize}
	\item the vertices in $X_i$ that are in the partial solution~$D$;
	\item for every vertex in $X_i$ (both in $D$ and not in $D$), the number of neighbours in~$D$.
\end{itemize}
This corresponds exactly to the bookkeeping required to verify whether a partial solution locally satisfies the requirements imposed by the specific $[\sigma,\rho]$-domination problem.
As such, we can identify every equivalence class using an assignment of \emph{labels} (sometimes also called \emph{states}) that capture the above properties to the vertices in $X_i$: such an assignment is called a \emph{state colouring}.
Given $\sigma, \rho \subseteq \Nat$, define the set of labels $C = C_\sigma \cup C_\rho$ as follows (the meaning of a label is explained below): 
\begin{align*}
	C_\sigma &= \left\{ 
	\begin{array}{lll}
		\{ |0|_\sigma, |1|_\sigma, |2|_\sigma, \ldots, |\ell - 1|_\sigma, |\ell|_\sigma \} & \textrm{ if $\sigma$ finite} & \textrm{ where $\ell = \max\{\sigma\}$}\\
		\{|\!\geq\!0|_\sigma\} & \textrm{ if $\sigma = \Nat$}& \\
		\{ |0|_\sigma, |1|_\sigma, |2|_\sigma, \ldots, |\ell-1|_\sigma, |\!\geq\!\ell|_\sigma \} & \textrm{ if $\sigma\neq\Nat$ cofinite} & \textrm{ where $\ell = \max\{\Nat\setminus\sigma\}+1$}
	\end{array}\right. \\
	C_\rho &= \left\{ 
	\begin{array}{lll}
		\{ |0|_\rho, |1|_\rho, |2|_\rho, \ldots, |\ell - 1|_\rho, |\ell|_\rho \} & \textrm{ if $\rho$ finite} & \textrm{ where $\ell = \max\{\rho\}$}\\
		\{|\!\geq\!0|_\rho\} & \textrm{ if $\rho = \Nat$} & \\
		\{ |0|_\rho, |1|_\rho, |2|_\rho, \ldots, |\ell-1|_\rho, |\!\geq\!\ell|_\rho \} & \textrm{ if $\rho\neq\Nat$ cofinite} & \textrm{ where $\ell = \max\{\Nat\setminus\rho\}+1 $}	
	\end{array}\right.
\end{align*}
We will use the $||_\rho$ and $||_\sigma$ notation to denote labels from $C_\rho$, respectively $C_\sigma$.
In general, when we write $|l|_\rho$ or $|l|_\sigma$, with a variable~$l$, we mean the labels that are not equal to $|\!\geq\!\ell|_\rho$ or $|\!\geq\!\ell|_\sigma$.
This allows us to refer to other labels by expressions such as $|l-1|_\rho$.
The symbol~$\ell$ is reserved to indicate the last labels $|\ell|_\sigma$, $|\!\geq\!\ell|_\sigma$, $|\ell|_\rho$, $|\!\geq\!\ell|_\rho$ and is used similarly to form labels such as $|\ell-1|_\rho$.

Let $C^{X_i}$ be the set of assignments of labels from~$C$ to the vertices in~$X_i$.
A label from $C_\sigma$ for a vertex $v \in X_i$ indicates that~$v$ is in the solution set~$D$ in the partial solution, a label from $C_\rho$ indicates that~$v$ is not.
Furthermore, the numbers in the labels indicate the number of neighbours that~$v$ has in $D$; the $\geq$ symbol in the label $|\!\geq\!1|_\rho$ indicates that~$v$ has this number of neighbours (one in this case) in $D$ or more.
For an example, consider {\sc Minimum Dominating Set} for which $\sigma = \Nat$ and $\rho = \Nat \setminus \{0\}$; for this problem $C = \{|\!\geq\!0|_\sigma, |0|_\rho, |\!\geq\!1|_\rho \}$.

Now, the elements from $C^{X_i}$ bijectively correspond to the above defined equivalence classes of partial solutions on $G_i$.
Consequently, we can index the memoisation table $A_i$ by $C^{X_i}$.
To keep the dynamic programming recurrences in this paper simple, we will not store partial solutions in $A_i$, only the required partial solution values or counts.
That is, from here on, let the table $A_i$ be a function $A_i : C^{X_i} \rightarrow \{0,1,..,M\} \cup \{\infty\}$ that assigns a number to each equivalence class of partial solutions.
In an existence variant of a problem, we let $A_i(\vec{c})$, for $\vec{c} \in C^{X_i}$, be 0 or 1 indicating whether a partial solution of this equivalence class exists.
In an optimisation variant, $A_i(\vec{c})$ indicates the size of a dominating partial solution in this equivalence class, or $\infty$ if no such partial solution exists.
For convenience reasons\footnote{In this way, we do not have to correct for double counting in join nodes in the rest of this paper.}, we let $A_i(\vec{c})$, for $\vec{c} \in C^{X_i}$, contain the size of the partial solution~$D'$ restricted to $V' \setminus X'$, i.e., the size of a corresponding partial solution equals $A_i(\vec{c})$ plus the number of $\sigma$ labels in $\vec{c}$.
In a counting variant, $A_i(\vec{c})$ indicates the number of partial solutions in the equivalence class of $\vec{c}$.
Notice that for an existence variant, we can bound $M$ by $1$; for an optimisation variant, we can bound $M$ by $n$; and for a counting variant, we can bound $M$ by $2^n$.

Below, we give explicit recurrences for $A_i$ for solving a minimisation variant of a $[\sigma,\rho]$-dominating problem by dynamic programming on a nice tree decomposition~$T$.
Modifying the recurrences to the existence or counting variant of the problem is an easy exercise.
Extensions to the recurrences in which partial solutions are stored (for existence and optimisation variants) are easy to make, but tedious to write down formally.
This is also to true for the extension to the optimisation counting variant where one needs to keep track of both the size and the number of such partial solutions.

\paragraph{Leaf node.}
Let $i$ be a leaf node of~$T$.
Since $X_i = \emptyset$, the only partial solution is $\emptyset$ with size zero: this size is stored for the empty vector $[]$.
\[
	A_i( [] ) = 0
\]

\paragraph{Introduce node.}
Let $i$ be an introduce node of~$T$ with child node~$j$.
Let $X_i = X_j \cup \{v\}$ for some $v \in V \setminus V_j$.
For $\vec{c} \in C^{X_j}$ and $c_v \in C$ the label for vertex~$v$ denote by $[\vec{c}, c_v]$ the vector $\vec{c}$ with the element $c_v$ appended to it such that $[\vec{c}, c_v] \in C^{X_i}$.
Now:
\[ 
	A_i( [\vec{c}, c_v] ) = \left\{ \begin{array}{ll} 
		A_j([\vec{c}]) & \textrm{ if $c_v \in \{|0|_\sigma,|\!\geq\!0|_\sigma\}$ or $c_v \in  \{|0|_\rho,|\!\geq\!0|_\rho\}$} \\ 
		\infty & \textrm{ otherwise}
	\end{array}\right. 
\]
Here, $G_i$ equals $G_j$ with one added isolated vertex~$v$.
Hence, $v$ can be in the partial solution or not, and both choices do not influence the partial solution size on $V_i \setminus X_i$ (which equals $V_j \setminus X_j$).
Note that only one of the labels from $\{|0|_\sigma,|\!\geq\!0|_\sigma\}$ and one from $\{|0|_\rho,|\!\geq\!0|_\rho\}$ is used, and which depends on the specific $[\sigma,\rho]$-domination problem that we are solving.

\paragraph{Forget node.}
Let $i$ be a forget node of~$T$ with child node~$j$.
Let $X_i = X_j \setminus \{v\}$ for some $v \in X_j$.

By definition of $G_i$, $G_i$ contains edges between $v$ and vertices in $X_i$ while $G_j$ does not.
To account for these edges, we start by updating the given table $A_j$ such that it accounts for the additional edges: that is, for an edge $\{u,v\}$ with $u \in X_i$, we adjust the counts of the number of neighbours expressed in the state colourings for $u$ and~$v$.
We do so before we construct table $A_i$.

Let $[\vec{c},c_u,c_v] \in C^{X_j}$ be such that $c_u$ and $c_v$ are labels for $u$ and $v$ respectively.
For every edge $\{u,v\}$ with $u \in X_i$, we update $A_j$ twice, once for $u$ and once for $v$.
We update $A_j$ for $u$ as follows:
\[ 
	A_j([\vec{c},c_u,c_v]) \!:=\! \left\{\begin{array}{ll}
		A_j([\vec{c},c_u,c_v]) & \textrm{ if $c_v \in C_\rho$} \\
		\infty & \textrm{ if $c_v \in C_\sigma$, $c_u \in \{|0|_\rho, |0|_\sigma\}$} \\
		A_j([\vec{c},|l - 1|_\rho,c_v]) & \textrm{ if $c_v \in C_\sigma$, $c_u = |l|_\rho$, $l > 0$} \\
		A_j([\vec{c},|l - 1|_\sigma,c_v]) & \textrm{ if $c_v \in C_\sigma$, $c_u = |l|_\sigma$, $l > 0$} \\
		\min\{ A_j([\vec{c},|\ell - 1|_\rho,c_v]), A_j([\vec{c},|\!\geq\!\ell|_\rho,c_v]) \} & \textrm{ if $c_v \in C_\sigma$, $c_u = |\!\geq\!\ell|_\rho$} \\
		\min\{ A_j([\vec{c},|\ell - 1|_\sigma,c_v]), A_j([\vec{c},|\!\geq\!\ell|_\sigma,c_v]) \} & \textrm{ if $c_v \in C_\sigma$, $c_u = |\!\geq\!\ell|_\sigma$} \\
	\end{array}\right.  
\]
No update needs to be done if $v$ is not in the partial solution~$D$ (first line).
If $c_u$ indicates that $u$ has no neighbours in~$D$ while $v \in D$, then no such partial solution exists (second line).
Otherwise, the counts in the label of $u$ need to account for the extra neighbour.
In the last four lines, we perform the required label update for all other labels giving special attention to the case where a $|\!\geq\!\ell|_\sigma$ or $|\!\geq\!\ell|_\rho$ label is used.
Here, the minimum needs to be taken over two equivalence classes that through the added edge become equivalent: we take the minimum because we are solving the minimisation variant.
Updating $A_j$ for $v$ goes identically with the roles of $u$ and $v$ switched, and as stated above, we perform this update for all edges incident to $v$ in $G_j$.

Next, we compute $A_i$ and start keeping track of equivalence classes based on $X_i$ instead of based on $X_j$.
To do so, we select a dominating solution from the partial solution equivalence classes for which $v$ has a number of neighbours in~$D$ that corresponds to the specific $[\sigma,\rho]$-domination problem:
\[
	A_i(\vec{c}) = \min_{\textrm{$c_v$ a valid label} } A_j([\vec{c},c_v])
\]
Here, a valid label $c_v$ is any label that corresponds to having the correct number of neighbours in~$D$ as defined by the specific $[\sigma,\rho]$-domination problem: $c_v$ is a label $|l|_\sigma$ or $|l|_\rho$ for which $l \in \sigma$ or $l \in \rho$, respectively, or $c_v$ is a label $|\!\geq\!\ell|_\rho$ or $|\!\geq\!\ell|_\sigma$ in case of cofinite $\sigma$ or $\rho$.

\paragraph{Join node.}
Let $A_i$ be the memoisation table for a join node $i$ of $T$ with child nodes $l$ and $r$.
Here we give a simple algorithm for the join node; in Section~\ref{sec:joins}, we survey more involved approaches.

A trivial algorithm to compute $A_i$ would loop over all pairs of state colourings $\vec{c}_l$, $\vec{c}_r$ of $X_i$ that agree on which vertices are in the solution set $D$, and then consider two corresponding partial solutions $D_l$ on $G_l$ and $D_r$ on $G_r$ and infer the state colouring $\vec{c}_i$ of the partial solution $D_l \cup D_r$ on $G_i$.
It then stores in $A_i$ the minimum size of a solution for each equivalence class for~$G_i$.

Note that the agreement on which vertices are in $D$ is necessary for $D_l \cup D_r$ to be a valid partial solution: otherwise vertices that are no longer in $X_i$ can obtain additional neighbours in $D$.
At the same time the agreement is not a too tight restriction as any partial solution~$D$ on $G_i$ can trivially be decomposed into partial solutions on $G_l$ and $G_r$ that agree on which vertices on $X_i$ are in $D$.

\paragraph{Root node.}
In the root node~$r$ of $T$ (which is a forget node), $X_r = \emptyset$, $G_r = G$ and consequently $A_r([])$ is the minimum size of a $[\sigma,\rho]$-dominating set on $G$.
The result we set out to compute!

\begin{lemma} \label{lem:bottleneck}
	Let $\mathcal{P}$ be the minimisation variant of a $[\sigma,\rho]$-domination problem with label set $C$ using $s=|C|$ labels.
	Let $\mathcal{A}$ be an algorithm for the computations in a join node for problem~$\mathcal{P}$ that, given a join node~$i$ with $|X_i|=k$ and the memoisation tables $A_l$ and $A_r$ for its child nodes, computes the memoisation table $A_i$ in $\BigO(f(n,k))$ arithmetic operations.
	Then, given a graph $G$ with a tree decomposition~$T$ of width $t$, $\mathcal{P}$ can be solved on $G$ in $\BigO( ( s^{t+1}t + f(n,t+1))n )$ arithmetic operations.
\end{lemma}
\begin{proof}
	First transform~$T$ into a nice tree decomposition~$T'$ with $\BigO(n)$ nodes.
	If we show that the table $A_j$ associated to any node~$j$ of~$T'$ can be computed in $\BigO( s^kk + f(n,k) )$ arithmetic operations, then the result follows as $k \leq t + 1$.
	Consider the recurrences in the dynamic programming algorithm exposed above.
	The result trivially holds for leaf and root nodes, and also for the join nodes by definition of $\mathcal{A}$.
	It is easy to see that in the recurrences for the introduce and forget nodes, every value is computed using a constant amount of work.
	Since the tables are of size $s^k$, and for a forget node we need to do at most $k$ update steps as we can add at most $k$ edges, the result follows.
\qed
\end{proof}

It is not difficult to modify the above algorithm to obtain:
\begin{proposition} \label{prop:bottleneck2}
	Lemma~\ref{lem:bottleneck} holds irrespective of $\mathcal{P}$ being a existence, maximisation, minimisation, counting, counting minimisation or counting maximisation variant of a $[\sigma,\rho]$-domination problem.
\end{proposition}

\section{Overview of Fast Transforms} \label{sec:transforms}
To obtain fast algorithms for the computations in join nodes of a nice tree decomposition, we use several well-known algebraic transforms, specifically the M\"obius transform and the Fourier transform.
We opted for a reasonably extensive coverage of this standard material because of completeness reasons and because the details matter for some of the arguments in Section~\ref{sec:joins} and~\ref{sec:new}.

Recall that, in the introduction on dynamic programming for $[\sigma,\rho]$-domination problems, we stored integers in the domain $\{0,1,\ldots,M\}$ for some large integer $M$.
We present the algebraic transforms using computations in $\Field_p$, the field of integers modulo a prime number $p$.
Since we know that, for a join node $i$ with child nodes $l$, $r$, all values in the memoisation tables $A_i$, $A_l$ and $A_r$ are in $\{0,1,\ldots,M\}$, we can do the computations in $\Field_p$ as long as $p > M$.

In the literature, the discrete Fourier transform is often defined on sequences in~$\mathbb{C}$.
We choose $\Field_p$ to avoid any analysis of rounding errors, especially when we combine it with the use of zeta and M\"obius transforms.
Using $\Field_p$ does require that $p$ is chosen appropriately: $\Field_p$ must contain certain roots of unit required for the Fourier transforms.
In the statements of definitions, propositions and lemmas in this section, we will sometimes say that $p$ is \emph{chosen appropriately} to state that $\Field_p$ contains the roots of unity required in the definition or in the following proof.
A short discussion on how to choose a proper prime number~$p$ such that this condition is satisfied can be found in Section~\ref{sec:choosep}.

\subsection{The Discrete Fourier Transforms Using Modular Arithmetic} \label{sec:fourier}
\begin{definition}[discrete Fourier transform] \label{def:dft}
	Let $\vec{a}=(a_i)_{i=0}^{r-1}$ be a sequence of numbers in $\Field_p$, and let $\omega_r$ be an $r$-th root of unity in $\Field_p$.
	The \emph{discrete Fourier transform} and \emph{inverse discrete Fourier transform} are transformations between sequences of length~$r$ in $\Field_p$ defined as follows:
	\begin{align*}
	\DFT(\vec{a})_i = \sum_{j=0}^{r-1} \omega_r^{ij} a_j && \InvDFT(\vec{a})_i = \frac{1}{r} \sum_{j=0}^{r-1} \omega_r^{-ij} a_j
	\end{align*}
\end{definition}
Recall that an $r$-th root of unity is an element $x \in \Field_p$ such that $x^r = 1$ while $x^l \not= 1$ for all $l < 1$.

These two transformations are inverses as their names suggest.
\begin{proposition} \label{prop:invdft}
$\InvDFT( \DFT( \vec{a} ))_i = a_i$
\end{proposition}
\begin{proof}
	In the derivation below, we first fill in the definitions and rearrange the terms (\ref{eq:invdft1}).
	Then, we split the sum based on $k=i$ and $k \neq i$ (from \ref{eq:invdft1} to \ref{eq:invdft2}).
	\begin{align}
	\InvDFT( \DFT( \vec{a} ))_i 
	&= \frac{1}{r} \sum_{j=0}^{r-1} \omega_r^{-ij} \sum_{k=0}^{r-1} \omega_r^{jk} a_k 
	= \frac{1}{r} \sum_{k=0}^{r-1} a_k \sum_{j=0}^{r-1} (\omega_r^{k-i})^j \label{eq:invdft1} \\
	&= \frac{1}{r} a_i \sum_{j=0}^{r-1} (\omega_r^{i-i})^j \; + \; \frac{1}{r} \sum_{\substack{k=0 \\ k \neq i}}^{r-1} a_k \sum_{j=0}^{r-1} (\omega_r^{k-i})^j \label{eq:invdft2} = a_i \frac{1}{r} \sum_{j=0}^{r-1} 1 \; + \; \frac{1}{r} \sum_{\substack{k=0 \\ k \neq i}}^{r-1} a_k \cdot 0
	= a_i 
	\end{align}
	Finally, we use that the first part of the sum is trivial as $\omega_r^{i-i}=\omega_r^0=1$, while the second part cancels as $\sum_{j=0}^{r-1} (\omega_r^{k-i})^j = \frac{1-\omega_r^{(k-i)r}}{1-\omega_r^{k-i}}$ is a geometric series with $\omega_r^{(k-i)r}=(\omega_r^r)^{k-i}=1^{k-i}=1$.
	\qed
\end{proof}

There exist fast algorithms for the discrete Fourier transform and its inverse, called fast Fourier transforms (FFT's), e.g., see the Cooley-Tukey FFT algorithm~\cite{CooleyT65} and Rader's FFT algorithm~\cite{Rader68}.
These algorithms are not particularly difficult to understand, but beyond the scope of this paper.
\begin{proposition}[fast Fourier transform] \label{prop:fft}
	The discrete Fourier transform and its inverse for sequences of length~$r$ can be computed in $\BigO(r \log r)$ arithmetic operations.
\end{proposition}

The definition of the discrete Fourier transform can be naturally extended from sequences to higher dimensional structures.
Let $\Int_r$ be the commutative ring of integers modulo~$r$ (here the modulus~$r$ can be non-prime), and let $\Int_r^k$ be the $\Int_r$-module of $k$-tuples with elements from $\Int_r$.
\begin{definition}[multidimensional discrete Fourier transform] \label{def:multidft}
	Let $Z = \Int_{r_1} \!\times \Int_{r_2} \!\times \cdots \times \Int_{r_k}$, and let $R=\prod_{i=1}^k r_i$.
	Also, let $\tens{A} \!=\! (a_{\vec{x}})_{\vec{x} \in Z}$ be a tensor of rank~$k$ with elements in $\Field_p$ indexed by the $k$-tuple $\vec{x} = [x_1,x_2,\ldots,x_k]$, where $p$ is chosen appropriately.
	The \emph{multidimensional discrete Fourier transform} and \emph{inverse multidimensional discrete Fourier transform} are defined as follows:
	\begin{align*}
	\DFT_k(\tens{A})_{\vec{x}} &= \sum_{y_1=0}^{r_1-1} \omega_{r_1}^{x_1y_1} \sum_{y_2=0}^{r_2-1} \omega_{r_2}^{x_2y_2} \cdots  \sum_{y_k=0}^{r_k-1} \omega_{r_k}^{x_ky_k} a_{\vec{y}} \\
	\InvDFT_k(\tens{A})_{\vec{x}} &= \frac{1}{R} \sum_{y_1=0}^{r_1-1} \omega_{r_1}^{-x_1y_1} \sum_{y_2=0}^{r_2-1} \omega_{r_2}^{-x_2y_2} \cdots  \sum_{y_k=0}^{r_k-1} \omega_{r_k}^{-x_ky_k} a_{\vec{y}}
	\end{align*}
	When $r = r_1 = r_2 = \ldots = r_k$, this simplifies to the following:
	\begin{align*}
	\DFT_k(\tens{A})_{\vec{x}} = \sum_{\vec{y} \in \Int_r^k} \omega_r^{\vec{x} \cdot \vec{y}} a_{\vec{y}} &&
	\InvDFT_k(\tens{A})_{\vec{x}} = \frac{1}{r^k} \sum_{\vec{y} \in \Int_r^k} \omega_r^{-\vec{x} \cdot \vec{y}} a_{\vec{y}}
	\end{align*}
	where the expressions in the exponents are the dot products on the tuples $\vec{x}$ and $\vec{y}$ in $\Int_r^k$.
\end{definition}
Note that the dot products are in exponents of which the base is an $r$-th root of unity, hence they are computed modulo~$r$: this agrees with the notation where $\vec{x}$ and $\vec{y}$ are taken from $\Int_r^k$.

\begin{proposition}[fast multidimensional discrete Fourier transform] \label{prop:multidimfft}
	Let $Z = \Int_{r_1} \!\times \Int_{r_2} \!\times \cdots \times \Int_{r_k}$, and let $R=\prod_{i=1}^k r_i$.
	Also, let $\tens{A}$ be a tensor of rank $k$ with elements in $\Field_p$, $\tens{A} \!=\! (a_{\vec{x}})_{\vec{x} \in Z}$, where $p$ is chosen appropriately.
	The multidimensional discrete Fourier transform and inverse multidimensional discrete Fourier transform of $\tens{A}$ can be computed in $\BigO(R \log(R))$ time.
\end{proposition}
\begin{proof}
	Denote by $\vec{x}[x_i \leftarrow y]$ the tuple $\vec{x}$ with the $i$-th coordinate of $\vec{x}$ replaced by $y$.
	We compute $DFT_k(\tens{A})$ with an algorithm that uses $k$-steps.
	Let $\tens{A}_0 = \tens{A}$.
	At the $i$-th step of the algorithm, let:
	\[ (\tens{A}_i)_{\vec{x}} = \sum_{j=0}^{r_i-1} \omega_{r_i}^{x_i j} a_{\vec{x}[x_i \leftarrow j]}  \]
	Notice that if $k=1$, this formula equals the one dimensional discrete Fourier transform.
	It is not hard to see that $\tens{A}_k$ is the $k$-dimensional Fourier transform of $\tens{A}$: if one repeatedly substitutes the formula for $\tens{A}_{i-1}$ in the formula for $\tens{A}_i$ starting at $i = k$, one obtains the (non-simplified) formula for the $k$-dimensional Fourier transform in Definition~\ref{def:multidft}.
	
	For the inverse multidimensional Fourier transform, almost the same procedure can be followed.
	Let $\tens{A}_0 = \tens{A}$ and use the following formula at the $i$-th step, finally obtaining the result $\tens{A}_k$.
	Here, again if $k=1$, this formula equals the one dimensional inverse discrete Fourier transform.
	\[ (\tens{A}_i)_{\vec{x}} = \frac{1}{r_i} \sum_{j=0}^{r_i-1} \omega_{r_i}^{-x_i j} a_{\vec{x}[x_i \leftarrow j]} \]
	
	For the running time, notice that step $i$ preforms $\frac{R}{r_i}$ standard 1-dimensional (inverse) discrete Fourier transforms on a sequence of length~$r_i$.
	By Proposition~\ref{prop:fft} this can be done in $\BigO( \frac{R}{r_i} r_i \log(r_i)) = \BigO(R \log(r_i))$ time.
	This leads to a total running time of $\BigO( R \sum_{i=1}^k \log(r_i) ) = \BigO(R \log(R))$.
	\qed
\end{proof}
In the above proof, the sequences $\tens{A}_0, \tens{A}_1, \ldots, \tens{A}_k$ are created using 1-dimensional (inverse) discrete Fourier transforms.
Because the 1-dimensional discrete Fourier transform and 1-dimensional inverse discrete Fourier transform are inverses, it directly follows that the sequence $\tens{A}_0, \tens{A}_1, \ldots, \tens{A}_k$ used in the $k$-dimensional discrete Fourier transform algorithm equals the sequence $\tens{A}_k, \tens{A}_{k-1}, \ldots, \tens{A}_0$ used in the inverse $k$-dimensional discrete Fourier transform.
I.e., as the name suggests, the $k$-dimensional inverse discrete Fourier transform is the inverse of the $k$-dimensional discrete Fourier transform.

We mainly use the multidimensional fast discrete Fourier transform in combination with the well-known convolution theorem.
\begin{lemma}[multidimensional convolution theorem] \label{lem:multidimconv}
	Let $Z = \Int_{r_1} \!\times \Int_{r_2} \!\times \cdots \times \Int_{r_k}$, and let 
	$\tens{A} \!=\! (a_{\vec{x}})_{\vec{x} \in Z}$, $\tens{B} \!=\! (b_{\vec{x}})_{\vec{x} \in Z}$ be tensors of rank~$k$ with elements in $\Field_p$, where $p$ is chosen appropriately.
	Let the tensor multiplication $\tens{A} \cdot \tens{B}$ be defined point wise, and let for $a_{\vec{x}}$ and $b_{\vec{y}}$ the sum $\vec{x} + \vec{y}$ be defined as the sum in $Z$ (coordinate-wise with the $i$-th coordinate modulo~$r_i$). 
	Then:
	\[ \InvDFT_k \left( \DFT_k( \tens{A} ) \cdot \DFT_k( \tens{B} ) \right)_{\vec{x}} = \sum_{\vec{z_1} +\vec{z_2} \equiv \vec{x}} a_{\vec{z_1}} b_{\vec{z_2}} \]
\end{lemma}
\begin{proof}
	We prove the lemma for the simplified case where $r=r_1=r_2=\ldots=r_k$ and hence $Z = \Int_r^k$, the more general case goes analogously but is notation-wise much more tedious as one needs to differentiate between multiple moduli and their corresponding roots of unity.
	
	The proof follows the same pattern as in Proposition~\ref{prop:invdft}.
	That is, we first fill in the definitions~(\ref{eq:rankconv1}) and rearrange the terms (\ref{eq:rankconv2}).
	Next (\ref{eq:rankconv3}), we observe that the sum over all $\vec{j} \in \Int_r^k$ can be written as a product of $k$ smaller sums, each involving but one coordinate of $\vec{j}$.	
	\begin{align}
	\InvDFT_k \! \left( \DFT_k( \tens{A} ) \! \cdot \! \DFT_k( \tens{B} ) \right)_{\vec{x}}
	&= \frac{1}{r^k} \sum_{\vec{y} \in \Int_r^k} \omega_r^{-\vec{x} \cdot \vec{y}} \left( \sum_{\vec{z_1} \in \Int_r^k} \omega_r^{\vec{y} \cdot \vec{z_1}} a_{\vec{z_1}} \right) \left( \sum_{\vec{z_2} \in \Int_r^k} \omega_r^{\vec{y} \cdot \vec{z_2}} b_{\vec{z_2}} \right) \label{eq:rankconv1} \\
	&= \frac{1}{r^k} \sum_{\vec{z_1},\vec{z_2} \in \Int_r^k} a_{\vec{z_1}} b_{\vec{z_2}} \sum_{\vec{y} \in \Int_r^k} \omega_r^{\vec{y} \cdot (\vec{z_1}+\vec{z_2}-\vec{x})} \label{eq:rankconv2} \\
	&= \frac{1}{r^k} \sum_{\vec{z_1},\vec{z_2} \in \Int_r^k} a_{\vec{z_1}} b_{\vec{z_2}} \prod_{i=1}^k \left( \sum_{j=0}^{r-1} \omega_r^{j ( (\vec{z_1})_i + (\vec{z_2})_i - x_i ) } \right) \label{eq:rankconv3}
	\end{align}
	Here, $x_i$, $(\vec{z_1})_i$ and $(\vec{z_2})_i$ are the $i$-th components of $\vec{x}$, $\vec{z_1}$ and $\vec{z_2}$ respectively.
	
	When $x_i \equiv (\vec{z_1})_i + (\vec{z_2})_i$ modulo~$r$ in the parenthesised sum of Equation~\ref{eq:rankconv3}, this sum becomes $\sum_{j=0}^{r-1} \omega_r^0$ and thus equals~$r$.
	Otherwise, when $x_i \not\equiv (\vec{z_1})_i + (\vec{z_2})_i$ the parenthesised sum is again a geometric series: $\sum_{j=0}^{r-1} (\omega_r^{(\vec{z_1})_i + (\vec{z_2})_i - x_i})^j$ that solves to $\frac{1-(\omega_r^{(\vec{z_1})_i + (\vec{z_2})_i - x_i})^r}{1-\omega_r^{(\vec{z_1})_i + (\vec{z_2})_i - x_i}} = 0$ as $(\omega_r^{(\vec{z_1})_i + (\vec{z_2})_i - x_i})^r = (\omega_r^r)^{(\vec{z_1})_i + (\vec{z_2})_i - x_i} = 1$ in the numerator.
	
	Continuing from (\ref{eq:rankconv3}), we obtain:
	\begin{align}
	\InvDFT_k \! \left( \DFT_k( \tens{A} ) \! \cdot \! \DFT_k( \tens{B} ) \right)_{\vec{x}}
	&= \frac{1}{r^k} \sum_{\vec{z_1},\vec{z_2} \in \Int_r^k} a_{\vec{z_1}} b_{\vec{z_2}} \prod_{i=1}^k r [x_i = (\vec{z_1})_i + (\vec{z_2})_i] \label{eq:rankconv4} \\
	&= \sum_{\vec{z_1},\vec{z_2} \in \Int_r^k} a_{\vec{z_1}} b_{\vec{z_2}} \prod_{i=1}^k [x_i = (\vec{z_1})_i + (\vec{z_2})_i] = \sum_{\vec{z_1} + \vec{z_2} \equiv \vec{x}} a_{\vec{z_1}} b_{\vec{z_2}} \label{eq:rankconv5}
	\end{align}
	completing the proof.
	\qed
\end{proof}

Taking all the above together, we finally obtain the following result.
To distinguish from $\Int_r$, let $\Nat_{<r} = \{0,1,\ldots,r-1\}$ be the integers up to $r$ with standard operators without modulus (operations for which the result of standard operations on $\Nat_{<r}$ is outside $\Nat_{<r}$ are considered undefined, i.e., $2+2$ is undefined in $\Nat_{<3}$).
\begin{lemma}[cyclic and non-cyclic convolution] \label{lem:combinedconv}
	Let $N = \Nat_{<q_1}\! \times \Nat_{<q_2}\! \times \cdots \times \Nat_{<q_l}$, and let $Q = \prod_{i=1}^l  q_i$.
	Let $Z = \Int_{r_1} \!\times \Int_{r_2} \!\times \cdots \times \Int_{r_k}$, and let $R = \prod_{i=1}^k  r_i$.
	Let $f,g: Z \times N \rightarrow \Field_p$, where $p$ is chosen appropriately.
	And, let $h: Z \times N \rightarrow \Field_p$ be the combined (partially cyclic and partially non-cyclic) convolution of $f$ and $g$ defined as:
	\[ 
	h(\vec{x},\vec{i}) = \sum_{\vec{y_1} + \vec{y_2} \equiv \vec{x}} \sum_{\vec{j_1} + \vec{j_2} = \vec{i}} f(\vec{y_1},\vec{j_1}) g(\vec{y_2},\vec{j_2})
	\]
	where the sum $\vec{y_1} + \vec{y_2} \equiv \vec{x}$ is evaluated component-wise modulo $r_i$ at coordinate~$i$ (sum in $Z$), and the sum $\vec{j_1} + \vec{j_2} = \vec{i}$ is evaluated component-wise \emph{without} modulus (sum in $N$).
	Then, the combined convolution $h$ can be computed in $\BigO( R \, Q \, 2^l (\log(R)+\log(Q)+l) )$ arithmetic operations.
\end{lemma}
\begin{proof}
	We reduce the problem to a standard multidimensional convolution (with modulus) by padding the input with zeroes.
	To be precise, let $Z' = \Int_{2q_1} \times \Int_{2q_2} \times \cdots \times \Int_{2q_l}$ ($N$ with for each coordinate twice as many values and with modulo additions), and let $f',g': Z \times Z' \rightarrow \Field_p$ be equal to $f$ and $g$ on the intersection of their domains (where $N$ is interpreted as subset of $Z'$ by interpreting each $\Nat_{<q_i}$ as subset of $\Int_{2q_i}$) and zero otherwise.
	Use Proposition~\ref{prop:multidimfft} and Lemma~\ref{lem:multidimconv} to compute the standard multidimensional convolution of $f'$ and $g'$.
	Because $Z \times Z'$ has $R Q 2^l$ elements, this requires $\BigO( R Q 2^l (\log(R)+\log(Q)+l) )$ arithmetic operations.
	Because the padded zeroes prevent the circular convolution effect, we can extract $h$ by taking the restriction of the result to $Z \times N$.
	\qed
\end{proof}
Different than for previous propositions and lemmas, we have more freedom in choosing the prime $p$ that is 'chosen appropriately' in the lemma above.
For the given proof, appropriate means that in $\Field_p$ all $r_i$-th roots of unity exists and all $2q_i$-th roots of unity exist.
However in our applications of the lemma, $l$ is often fixed.
This means that the running time does not change if we allow $Z' = \Int_{s_1} \times \Int_{s_2} \times \cdots \times \Int_{s_l}$ with, for all $i$, $2q_i \leq s_i \leq c q_i$ for a small constant $c$.
In other words, with respect to the different $q_i$, there must be some root of unity, but the order of this root of unity has a broad range in which it is acceptable for our results to be valid.

\begin{corollary}[multidimensional non-cyclic convolution] \label{cor:noncyclicconv}
	Let $N = \Nat_{<q_1}\! \times \Nat_{<q_2}\! \times \cdots \times \Nat_{<q_l}$, and let $Q = \prod_{i=1}^l q_i$.
	Let $f,g: N \rightarrow \Field_p$, where $p$ is chosen appropriately.
	Let $h: N \rightarrow \Field_p$ be the non-cyclic convolution of $f$ and $g$ defined as:
	\[ 
	h(\vec{i}) = \sum_{\vec{j_1} + \vec{j_2} = \vec{i}} f(\vec{j_1}) g(\vec{j_2})
	\]
	where the sum $\vec{j_1} + \vec{j_2} = \vec{i}$ is evaluated component-wise \emph{without} modulus (sum in $N$).
	Then, $h$ can be computed in $\BigO( Q \, 2^l (\log(Q)+l) )$ arithmetic operations.
\end{corollary}
\begin{proof}
	Direct consequence of Lemma~\ref{lem:combinedconv} with $k = 0$.
	\qed
\end{proof}

\subsection{M\"obius Inversion Using Fast Zeta and Fast M\"obius Transforms} \label{sec:mobius}
The zeta and M\"obius transforms apply to functions on partially ordered sets.

\begin{definition}[zeta and M\"obius transform]
	Let $P$ be a partially ordered set.
	Given a function $f: P \rightarrow \Field_p$, the \emph{zeta transform} $\zeta(f)$ and the \emph{M\"obius transform} $\mu(f)$ are defined as follows:
	\begin{align*}
		\zeta(f)(x) = \sum_{y \leq x} f(y) 
		&& \mu(f)(x) = \sum_{y \leq x} \mu(y,x) f(y) 
		&& \textrm{where}
		&& \mu(x,y)=\left\{ \begin{array}{ll} 
			1 & \textrm{ if $x = y$}\\
			- \sum_{x < z \leq y}\mu(z,y) & \textrm{ if $x < y$}
		\end{array} \right.
	\end{align*}
	The recursively defined function $\mu(x,y)$ on pairs $x,y \in P$ with $x \leq y$ is the M\"obius function of $P$.
\end{definition}

The \emph{zeta transform} $\zeta(f)$ and the \emph{M\"obius transform} are inverses, as we will now show.
\begin{lemma}[M\"obius inversion] \label{lem:mobiusinversion}
	Let $f: P \rightarrow \Field_p$ any function, then $\mu( \zeta( f ) )(x) = f(x)$.
\end{lemma}
\begin{proof}
	Let $x, y \in P$ and consider the sum $\sum_{x \leq z \leq y} \mu(z,y)$.
	If $x = y$, then this sum equals $\mu(x,x) = 1$.
	If $x < y$, then this sum equals $\mu(x,y) + \sum_{x < z \leq y} \mu(z,y) = 0$ by definition of $\mu(x,y)$.
	As such:
	\begin{align*}
		\mu( \zeta( f ) )(x) = \sum_{y \leq x} \mu(y,x) \sum_{z \leq y} f(z) = \sum_{z \leq x} f(z) \sum_{z \leq y \leq x} \mu(y,x) = \sum_{z \leq x}f(z) [z=x] = f(x)
	\end{align*}
	The first equality is by expanding the definitions.
	The second follows by reordering terms.
	And, the third follows from the above, where $[z=x]$ is Iverson notation that is $1$ if $z=x$ and $0$ otherwise.
\qed
\end{proof}
In this paper, we will not define any M\"obius transform explicitly.
We will show how to compute zeta transforms $\zeta(f)$ of functions~$f: P \rightarrow \Field_p$ for some partial orders $P$. 
Then, given $\zeta(f)$, we show that we can reconstruct $f$.
This reconstruction (implicitly) is an algorithm for the M\"obius transform because a consequence of Lemma~\ref{lem:mobiusinversion} is that the zeta transform has a unique inverse.

M\"obius inversion is often used in relation to lattices.
A \emph{meet-semilattice} is a partial order~$P$ on which, for any two elements in $x, y \in P$, the \emph{meet} $x \wedge y$ (greatest lower bound) is properly defined.
Similarly, a \emph{join-semilattice} is a partial order~$P$ set on which, for any two elements in $x, y \in P$, the \emph{join} $x \vee y$ (smallest upper bound) is properly defined.
A lattice is a partial order~$P$ that is both a meet and a join semi-lattice. 
An example is the finite lattice~$\Nat_{<r}^k$ with the coordinate-wise natural order and where the meet and join are the coordinate-wise minimum and maximum.

We will use M\"obius inversion on partial orders that are Cartesian products~$P^k$ of a smaller partial order $P$.
For $\vec{x}, \vec{y} \in P^k$, $\vec{x} = [x_1,x_2,\ldots,x_k]$, $\vec{y} = [y_1,y_2,\ldots,y_k]$, we write $\vec{x} \leq \vec{y}$ if and only if $x_i \leq y_i$ for all $i$.
Additionally, our partial orders have the property that for every $x \in P$, the downward closed set $\{y \in P | y \leq x\}$ forms a join-semilattice.
It is not hard to see that if for every $x \in P$, $\{y \in P | y \leq x\}$ forms a join-semilattice, then for every $\vec{x} \in P^k$, $\{\vec{y} \in P^k | \vec{y} \leq \vec{x}\}$ forms a join-semilattice as well, where the join operation is defined coordinate-wise.

On the subset lattice (isomorphic to $\Nat_{<2}^k$) there are well-known fast algorithms for the zeta and M\"obius transforms, often referred to as Yates' algorithm~\cite{Yates37}, see also~\cite{BjorklundHKK07,Kennes91}.
Below, we generalise these algorithms to partial orders $P^k$ for which, for every $\vec{x} \in P^k$, the set $\{\vec{y} \in P^k | \vec{y} \leq \vec{x}\}$ forms a join-semilattice.
For fast zeta and M\"obius transforms on arbitrary finite lattices, see~\cite{BjorklundHKKNP15}.

\begin{proposition}[fast zeta and M\"obius transforms] \label{prop:fastzeta}
	The zeta transform and M\"obius transform of a function $f:\Nat_{<r}^k \rightarrow \Field_p$ can be computed in $\BigO(r^kk)$ arithmetic operations.
\end{proposition}
\begin{proof}
	We compute $\zeta(f)$ with an algorithm that uses $k$ steps. 
	Let $f_0 = f$, and let $\vec{x} = [x_1,\ldots,x_k]$.
	Denote by $\vec{x}[x_i \leftarrow y]$ the tuple $\vec{x}$ with the value on the $i$-th coordinate replaced by $y$.
	At the $i$-th step of the algorithm, we compute $f_i$ recursively using the left formula below. 
	\begin{align}
		f_i(\vec{x}) = \left\{\begin{array}{ll} 
			f_{i-1}(\vec{x}) & \textrm{ if $x_i = 0$} \\
			f_i(\vec{x}[x_i \leftarrow x_i - 1]) + f_{i-1}(\vec{x}) & \textrm{ if $x_i > 0$}
		\end{array} \right. 
		&&
		f_i(\vec{x}) = \sum_{j \leq x_i} f_{i-1}(\vec{x}[x_i \leftarrow j]) \label{eq:fastzeta}
	\end{align}
	The right formula above follows by induction on the left recurrence. 
	By induction on the step number~$i$, one easily sees that $f_i$ satisfies the equation below, from which we can obtain $\zeta(f)$ since $f_k = \zeta(f)$.
	The result for $\zeta(f)$ follows because each step computes $r^k$ values, each in constant time.
	\begin{align*}
		f_i(\vec{x}) = \sum_{y_1 \leq x_1} \sum_{y_2 \leq x_2} \cdots \sum_{y_i \leq x_i} 	f([y_1,y_2,\ldots,y_i,x_{i+1},x_{i+2},\ldots,x_k]) 
	\end{align*}
	For $\mu(f)$, we use that $\mu(f)$ is the inverse of $\zeta(f)$: the sequence $f_0, f_1, \ldots, f_k$ used to compute $\zeta(f)$ from $f$ can computationally be inverted to compute $f$ from $\zeta(f)$.
	That is, let $f_k = \zeta(f)$, and let:
	\begin{align}
		f_i(\vec{x}) = \left\{\begin{array}{ll} 
			f_{i+1}(\vec{x}) - f_{i+1}(\vec{x}[x_i \leftarrow x_i - 1]) & \textrm{ if $x_i > 0$} \\
			f_{i+1}(\vec{x}) & \textrm{ if $x_i = 0$}
		\end{array} \right. \label{eq:fastmobius}
	\end{align}
	Assuming that the $f_i$ were computed using Equation~\ref{eq:fastzeta}, and substituting the right part of (\ref{eq:fastzeta}) into the case where $x_i>0$ in Equation~\ref{eq:fastmobius}, we see that (\ref{eq:fastmobius}) computes the inverse of Equation~\ref{eq:fastzeta}:
	\begin{align*}
		\left(\sum_{j \leq x_i} f_{i-1}(\vec{x}[x_i \leftarrow j])\right) - \left( \sum_{j \leq x_i-1} f_{i-1}(\vec{x}[x_i \leftarrow j]) \right) =  f_i(\vec{x})
	\end{align*}
	Hence, we can reconstruct $f_0 = f$ again.
	Since the run time is the same, we obtain the result.
\qed
\end{proof}

It is easy to generalise the inductive proof above and obtain:
\begin{lemma} \label{lem:fastgenericzeta}
	Given algorithms for the zeta and M\"obius transform for functions $f:P \rightarrow \Field_p$ that use $\BigO(|P|)$ arithmetic operations, there are algorithms for the zeta and M\"obius transform for functions $f:P^k \rightarrow \Field_p$ that require $\BigO(|P^k|k)$ arithmetic operations.
\end{lemma}

The application of the zeta and M\"obius transform that is important to us is the following lemma.
\begin{lemma}[generalised covering product] \label{lem:coverprod}
	Let $P$ be a finite partial order such that, for every $\vec{x} \in P^k$, the set $\{\vec{y} \in P^k | \vec{y} \leq \vec{x}\}$ forms a join-semilattice, and let $f, g: P^k \rightarrow \Field_p$.
	
	Define the \emph{generalised covering product} $h: P^k \rightarrow \Field_p$ of $f$ and $g$ through:
	\[ h(\vec{x}) = \sum_{\vec{y_1}\vee \vec{y_2} = \vec{x}} f(\vec{y_1}) \, g(\vec{y_2}) \]
	Then $\mu(\zeta(f) \cdot \zeta(g))(\vec{x}) = h(\vec{x})$, where the product $\zeta(f) \cdot \zeta(g)$ is defined by point-wise multiplication.
\end{lemma}
\begin{proof}
	We will prove that $(\zeta(f) \cdot \zeta(g))(\vec{x}) = \zeta(h)(\vec{x})$, then the result follows from Lemma~\ref{lem:mobiusinversion}.
	\begin{align} \label{eq:coverprod1}
		(\zeta(f) \cdot \zeta(g))(\vec{x}) = \left( \sum_{\vec{y} \leq \vec{x}} f(\vec{y}) \right) \left( \sum_{\vec{y} \leq \vec{x}} g(\vec{y}) \right) = \sum_{\vec{y_1},\vec{y_2} \leq \vec{x}} f(\vec{y_1})\, g(\vec{y_2}) 
	\end{align}
	Here, we first use the definition of the $\zeta$-transform and then work out all the product terms.
	The result equals $\zeta(h)(\vec{x})$ as we now show by working out the definition of $\zeta(h)(\vec{x})$.
	\begin{align} \label{eq:coverprod2}
		\zeta(h)(\vec{x}) = \sum_{\vec{z} \leq \vec{x}} \sum_{\vec{y_1} \vee \vec{y_2} = \vec{z}} f(\vec{y_1}) \, g(\vec{y_2}) = \sum_{\vec{y_1},\vec{y_2} \leq \vec{x}} f(\vec{y_1}) \, g(\vec{y_2}) 
	\end{align}
	For the last equality, we reorder terms using that for any two $\vec{y_1}, \vec{y_2} \leq \vec{x}$ there is a unique~$\vec{z}$ such that $\vec{y_1} \vee \vec{y_2} = \vec{z}$; this is well defined as the set $\{\vec{y} \in P^k | \vec{y} \leq \vec{x}\}$ forms a join-semilattice.
\qed
\end{proof}

As a direct result, we obtain a generalisation of the covering product from~\cite{BjorklundHKK07}.
\begin{corollary} \label{cor:mobiusjoin}
The generalised covering product for $f, g: \Nat_{<r}^k \rightarrow \Field_p$ defined in the statement of Lemma~\ref{lem:coverprod} can be computed in $\BigO(r^kk)$ arithmetic operations.
\end{corollary}
\begin{proof}
	Combine Lemma~\ref{lem:coverprod} with the fast evaluation algorithms of Proposition~\ref{prop:fastzeta}.
\qed
\end{proof}

We conclude this part on zeta and M\"obius transforms by a theorem on a combined covering product and convolution product that we will use in the sections to come.
\begin{theorem} \label{thrm:generalmobiousjoin}
	Let $P$ be a finite partial order where, for every $\vec{x} \in P^k$, the set $\{\vec{y} \in P^k | \vec{y} \leq \vec{x}\}$ forms a join-semilattice.
	Let $N = \Nat_{<q_1} \!\times \Nat_{<q_2} \!\times \cdots \times \Nat_{<q_l}$, and let $Q = \prod_{i=1}^l  q_i$.
	Let $f,g: P^k \times N \rightarrow \Field_p$, where $p$ is chosen appropriately.
	Define $h: P^k \times N \rightarrow \Field_p$ as follows:
	\[ 
		h(\vec{x},\vec{i}) = \sum_{\vec{y_1}\vee \vec{y_2} = \vec{x}} \sum_{\vec{j_1}+\vec{j_2}=\vec{i}} f(\vec{y_1},\vec{j_1}) \, g(\vec{y_2},\vec{j_2})
	\]
	If $P$ allows zeta and M\"obius transforms using $\BigO(|P|)$ arithmetic operations for functions $f':P \rightarrow \Field_p$, then $h$ can be computed in $\BigO(|P|^k \, Q \, (k + \log(Q)))$ arithmetic operations.
	
	In particular, if $P^k = \Nat_{<r}^k$, then $h$ can be computed in $\BigO(r^k \, Q \, (k + \log(Q)))$ arithmetic operations.
\end{theorem}
\begin{proof}
	For the functions $f$, $g$, $h$, with domains $P^k \times N$, we write $\zeta(f(-,\vec{i}))(\vec{x}) = \sum_{\vec{y}\leq\vec{x}} f(\vec{y},\vec{i})$ to fix the second component when using the zeta transform.
	Following the same reasoning as in Equations~\ref{eq:coverprod2} and Equation~\ref{eq:coverprod1} in the proof of Lemma~\ref{lem:coverprod}, one easily obtains:
	\begin{align*} 
		\zeta(h(-,\vec{i}))(\vec{x}) 
		& = \sum_{\vec{y} \leq \vec{x}} \sum_{\vec{z_1}\vee \vec{z_2} = \vec{y}} \sum_{\vec{j_1}+\vec{j_2}=\vec{i}} f(\vec{z_1},\vec{j_1}) \, g(\vec{z_2},\vec{j_2})
		= \sum_{\vec{j_1}+\vec{j_2}=\vec{i}} \!\left( \sum_{\vec{y} \leq \vec{x}} \sum_{\vec{z_1}\vee \vec{z_2} = \vec{y}} f(\vec{z_1},\vec{j_1}) \, g(\vec{z_2},\vec{j_2}) \right) \\
		&= \sum_{\vec{j_1}+\vec{j_2}=\vec{i}} \!\left( \sum_{\vec{z_1}, \vec{z_2} \leq x} f(\vec{z_1},\vec{j_1}) \, g(\vec{z_2},\vec{j_2}) \right) 
		= \sum_{\vec{j_1}+\vec{j_2}=\vec{i}} \left( \zeta(f(-,\vec{j_1}) )\cdot \zeta(g(-,\vec{j_2})) \right)(\vec{x})
	\end{align*}
	Consequently, we can compute $h$ by evaluating this expression and taking the M\"obius transform.
	
	That is, we can compute $h$ by taking the following steps:
	\begin{enumerate}
		\item Compute a fast zeta transform of $f(-,\vec{j})$ and $g(-,\vec{j})$, for each fixed $\vec{j}\in N$ in $\BigO(|P|^kk)$ arithmetic operations using Lemma~\ref{lem:fastgenericzeta}.
		This takes $\BigO(|P|^k Q k)$ arithmetic operations in total.
		For each $\vec{j_1}, \vec{j_2} \in N$ and $\vec{x} \in P^k$, we now have $\zeta(f(-,\vec{j_1}))(\vec{x})$ and $\zeta(g(-,\vec{j_2}))(\vec{x})$.
		\item For each fixed $\vec{x} \in P^k$, compute the sum over all $\vec{j_1} + \vec{j_2} = \vec{i}$ using the fast convolution algorithm of Corollary~\ref{cor:noncyclicconv} in $\BigO(Q \log(Q))$ arithmetic operations.
		This takes $\BigO(|P|^k Q \log(Q))$ arithmetic operations in total.
		For each $\vec{x} \in P^k$ and $\vec{i}\in N$, we now have $\zeta(h(-,\vec{i}))(\vec{x})$.
		\item Finally, for each fixed $\vec{i}$, take the M\"obius transform of $\zeta(h(-,\vec{i}))(\vec{x})$ in $\BigO(|P|^kk)$ time using Lemma~\ref{lem:fastgenericzeta}.
		Like the first step, this takes $\BigO(|P|^k Q k)$ arithmetic operations in total.
		As a result, we have the required values $h(\vec{x},\vec{i})$.
	\end{enumerate}
	The running time follows by summing the times required for each of the three steps.
\qed
\end{proof}

\subsection{Modular Arithmetic.} \label{sec:choosep}
As discussed in the introduction to this section, we embed the integers $\{0,1,\ldots,M\}$ in the larger field $\Field_p$, for a prime $p > M$.
However, we need to choose $p$ 'appropriately' such that the resulting field~$\Field_p$ has the required root(s) of unity.
Below we give a short description of how this can be done.

Let $r_1, r_2, \ldots, r_k$ be distinct integers.
We look for a prime number~$p$ such that $\Field_p$ contains, for all~$i$, an $r_i$-th root of unity.
To find the prime $p$, we consider candidates $m_j = 1 + j R$, where $R = \prod_{i=1}^k r_i$, for $j$ large enough such that $m_j > M$.
By the prime number theorem for arithmetic progressions, the sequence $(m_j)_{j=1}^\infty$ contains $\BigO(\frac{1}{\phi(R)} \frac{x}{\ln(x)})$ prime numbers less than $x$, where $\phi$ is Euler's totient function.
Since prime testing can be done in polynomial time, we can look for the first candidate $m_j > M$ that is prime and choose $p$ as such.

By Euler's theorem, for any $x \in \Field_p$, with $p$ chosen as in the previous paragraph: $1 = x^{\phi(p)} = x^{p-1} = x^{jR}$.
As such, for any $x \in \Field_p$, $x^l$ with $l = \frac{jR}{r_i}$ is an $r_i$-th root of unity if $(x^l)^i \not= 1$ for all $i < r_i$.
Finding an appropriate $x$ is not difficult for small $r_i$ as an $\frac{1}{r_i}$-th fraction of all elements $x \in \Field_p$ results in $x^l$ being an $r_i$-th root of unity.
To see this, consider a generator $g$ of the multiplicative subgroup of $\Field_p$.
The sequence $g^1,g^2,\ldots,g^{p-1}$ equals all elements in $\Field_p \setminus \{0\}$.
Putting this sequence to the power~$l$ gives $g^l, g^{2l}, \ldots, g^{(p-1)l}$ which, by choice of~$l$, equals $\omega_{r_i}^1, \omega_{r_i}^2, \ldots, \omega_{r_i}^{(p-1)}$, where $\omega_{r_i}$ is an $r_i$-th root of unity in $\Field_p$.
Clearly, this forms $l$ times the sequence $\omega_{r_i}^1, \omega_{r_i}^2, \ldots, \omega_{r_i}^{r_i}$, as  $\omega_{r_i}^{r_i} = 1$.

\section{Fast Join Operations} \label{sec:joins}
Having introduced the basics of dynamic programming on tree decompositions for $[\sigma,\rho]$-domination problems in Section~\ref{sec:prelim}, and the basics of the fast Fourier and fast M\"obius transforms in Section~\ref{sec:transforms}, we are now ready to apply the fast transforms to obtain fast join operations.
We will survey some known techniques based on both transforms.

Recall that, to compute the memoisation table for a join node $i$ of a nice tree decomposition, we are given two memoisation tables $A_l$ and $A_r$ corresponding to the (left and right) child nodes of $i$.
These tables store a number for each state colouring $\vec{c}$ with labels from $C$ (as defined in Section~\ref{sec:sigmarho}): this number indicates the existence (0 or 1) of a, the size of a, and/or the number of (minimum/maximum size) partial solution(s) on $G_l$ and $G_r$ for the partial-solution equivalence class corresponding to~$\vec{c}$.
Here, a partial-solution equivalence class is uniquely identified by~$\vec{c}$ in the following way: a label for a vertex~$v$ in $\vec{c}$ defines whether $v$ is in the solution set or not, and it defines how many neighbours $v$ has in the solution set (see Section~\ref{sec:sigmarho}).
Our goal is to compute $A_i$.

As stated in Section~\ref{sec:sigmarho}, a trivial algorithm would loop over all combinations of state colourings $\vec{c}_l$, $\vec{c}_r$ of $X_i$ that agree on which vertices are in the solution set $D$.
Then, the algorithm considers two corresponding partial solutions $D_l$ on $G_l$ and $D_r$ on $G_r$, and it infers the state colouring $\vec{c}_i$ of the partial solution $D_l \cup D_r$ on $G_i$.
Over all constructed partial solutions on $G_i$, it stores in $A_i$ the minimum size of a solution for each equivalence class~$\vec{c}_i$.
It is not hard to show that one does not need to consider the partial solutions representing an equivalence class: given the state colourings $\vec{c}_l$ and $\vec{c}_r$, the state colouring of $D_l \cup D_r$ can be inferred directly.
This is done as follows.
Since $G_l$, $G_r$ and $G_i = G_l \oplus G_r$ do not contain edges between vertices in $X_i$, for any vertex $v \in X_i$, the number of neighbours in~$D_l$ and $D_r$ add up to the resulting number in $D_r \cup D_l$.
As such, for any vertex $v$, if $v$ has label $|l|_\sigma$ in $\vec{c}_l$ and $|l'|_\sigma$ in $\vec{c}_l$, then any combined partial solution has label $|l + l'|_\sigma$ for $v$ in the state colouring that identifies the equivalence class (or label $|\!\geq\!\ell|_\sigma$ when $l + l' \geq \ell$ and the $|\!\geq\!\ell|_\sigma$ label is in $C_\sigma$).
The same holds for labels $|l|_\rho$ from $C_\rho$.

We find it insightful to make tables, which we call `join tables', that visualise the resulting label of a vertex in $\vec{c}_i$ given its labels in $\vec{c}_l$ and $\vec{c}_r$: see Figure~\ref{fig:jointablesds}.
In these tables, the patterns emerge that our fast join operations must fulfil.
Here, one can see the running time that a trivial algorithm uses to perform the join: every non-empty cell represents a combination from $A_l$ and $A_r$ that can be made on each vertex coordinate (each vertex in $X_i$).
As a result, this trivial algorithm performs the join in $\BigOs( x^k )$ time, where $k = |X_i|$ and $x$ is the number of non-empty cells in the join table.
\begin{figure}[tb]
	\begin{center}
	\scalebox{.9}{
		\begin{tabular}{c|ccc|} 
			& $|\!\geq\!0|_\sigma$ & $|0|_\rho$ & $|\!\geq\!1|_\rho$ \\ \hline 
			$|\!\geq\!0|_\sigma$ & $|\!\geq\!0|_\sigma$ & & \\  
			$|0|_\rho$ & & $|0|_\rho$ & $|\!\geq\!1|_\rho$ \\ 
			$|\!\geq\!1|_\rho$ & & $|\!\geq\!1|_\rho$ & $|\!\geq\!1|_\rho$ \\ \hline
		\end{tabular}
		\hspace{0.2cm}
		\begin{tabular}{c|ccc|} 
			& $|0|_\sigma$ & $|0|_\rho$ & $|1|_\rho$ \\ \hline 
			$|0|_\sigma$ & $|0|_\sigma$ & & \\  
			$|0|_\rho$ & & $|0|_\rho$ & $|1|_\rho$ \\ 
			$|1|_\rho$ & & $|1|_\rho$ & \\ \hline
		\end{tabular}
		\hspace{0.2cm}	
		\begin{tabular}{c|ccccc|} 
			& $|0|_\sigma$ & $|1|_\sigma$ & $|2|_\sigma$ & $|3|_\sigma$ & $|\!\geq\!0|_\rho$ \\ \hline 
			$|0|_\sigma$ & $|0|_\sigma$ & $|1|_\sigma$ & $|2|_\sigma$ & $|3|_\sigma$ & \\ 
			$|1|_\sigma$ & $|1|_\sigma$ & $|2|_\sigma$ & $|3|_\sigma$ & & \\
			$|2|_\sigma$ & $|2|_\sigma$ & $|3|_\sigma$ & & & \\
			$|3|_\sigma$ & $|3|_\sigma$ & & & & \\
			$|\!\geq\!0|_\rho$ & & & & & $|\!\geq\!0|_\rho$ \\ \hline
		\end{tabular}
		\hspace{0.2cm}	
		\begin{tabular}{c|cccc|} 
			& $|0|_\sigma$ & $|\!\geq\!1|_\sigma$ & $|0|_\rho$ & $|\!\geq\!1|_\rho$ \\ \hline 
			$|0|_\sigma$ & $|0|_\sigma$ & $|\!\geq\!1|_\sigma$ & & \\
			$|\!\geq\!1|_\sigma$ & $|\!\geq\!1|_\sigma$ & $|\!\geq\!1|_\sigma$ & & \\
			$|0|_\rho$ & & & $|0|_\rho$ & $|\!\geq\!1|_\rho$ \\
			$|\!\geq\!1|_\rho$ & & & $|\!\geq\!1|_\rho$ & $|\!\geq\!1|_\rho$ \\ \hline	
		\end{tabular}
	}\end{center}
	\caption{Join tables corresponding, from left to right, to {\sc Dominating Set}, {\sc Strong Stable Set}, {\sc 3-Regular Subgraph}, and {\sc Total Dominating Set}.}
	\label{fig:jointablesds}
\end{figure}

\subsection{M\"obius Tranforms for Dominating Set and Independent Dominating Set}
For the {\sc Dominating Set} problem, consider the first (leftmost) join table in Figure~\ref{fig:jointablesds}.
Notice that this is also the join table for the {\sc Independent Dominating Set} problem.

We first consider the (non-optimisation) counting variant of {\sc Dominating Set} or {\sc Independent Dominating Set}: the table entries $A_i(\vec{c})$ represent the number of partial solutions to (independent) dominating set in the equivalence class represented by~$\vec{c}$ (partial solutions of any size).
\begin{lemma}[based on~\cite{vanRooijBR09}] \label{lem:countds}
	The join for the counting variant of {\sc Dominating Set} can be computed in $\BigO(3^kk)$ arithmetic operations.
\end{lemma}
\begin{proof}
	We first loop over all $2^k$ subsets $X_\sigma \subseteq X_i$, and for each subset $X_\sigma$, we fix the labels of vertices in $X_\sigma$ in $\vec{c}_l$, $\vec{c}_r$ and $\vec{c}_i$ to $|\!\geq\!0|_\sigma$.
	We then consider the subproblem that remains using only the labels $|0|_\rho$ and $|\!\geq\!1|_\rho$. 
	Let $X' = X_i \setminus X_\sigma$ be the vertices without fixed label, let $k' = |X'|$, and let $A_i'$, $A_l'$, $A_r'$ be the memoisation tables $A_i$, $A_l$, and $A_r$ after fixing the vertices with label $|\!\geq\!0|_\sigma$, i.e., $A_i'$, $A_l'$, $A_r'$ are indexed by state colourings~$\vec{c}_l'$ and~$\vec{c}_r'$ on $X_i'$ using only $|0|_\rho$ and $|\!\geq\!1|_\rho$-labels.
	
	To compute the join, we now essentially want to take a coordinate-wise maximum of the state colourings~$\vec{c}_l'$ and~$\vec{c}_r'$ (identifying $|0|_\rho$ with $0$ and $|\!\geq\!1|_\rho$ with $1$) to obtain the resulting state colouring~$\vec{c}_i'$ on $X_i'$.
	That is, to compute $A_i'(\vec{c}_i')$, we want to efficiently evaluate the following formula:
	\begin{equation} \label{eq:countds}
		A_i'(\vec{c}_i') = \sum_{\vec{c}_l' \vee \vec{c}_r' = \vec{c}_i'} A_l'(\vec{c}_l') A_r'(\vec{c}_r')
	\end{equation}
	where $\vee$ is the above discussed coordinate-wise maximum (identifying $C_\rho$ with $\Nat_{<2}^k$).
	Observe that this corresponds exactly to the covering product, generalised in Lemma~\ref{lem:coverprod} with $P = \Nat_{<2}^k$ and $f = A_l'$, $g = A_r'$.
	Consequently, this join can be computed in $\BigO(2^{k'}k')$ arithmetic operations by Corollary~\ref{cor:mobiusjoin}.
	Summing up the running time over all $2^k$ subsets of fixed labels, we obtain a running time of:
	\[ \BigO \left( \sum_{X' \subseteq X_i} 2^{|X'|}|X'| \right) = \BigO \left( \sum_{k'=0}^k \binom{k}{k'} 2^{k'}k' \right) = \BigO \left( k \sum_{k'=0}^k \binom{k}{k'} 2^{k'} 1^{k-k'} \right) = \BigO( k (2 + 1)^k ) = \BigO( 3^kk ) \]
	where we group the subsets $X' =  X_i \setminus X_\sigma$ of the same size and then use the binomial theorem.
\qed
\end{proof}
\begin{corollary}
	Given a graph~$G$ with a tree decomposition~$T$ of $G$ of width~$t$, the number of (independent) dominating sets can be computed in $\BigO(3^ttn)$ arithmetic operations on $\BigO(n)$-bit numbers.
\end{corollary}
\begin{proof}
	Plug Lemma~\ref{lem:countds} into Lemma~\ref{prop:bottleneck2}.
	We can use $\BigO(n)$-bit numbers as the result is at most $2^n$.
\qed
\end{proof}

The above construction does not work directly for the minimisation version of {\sc (Independent) Dominating Set}, as then we are no longer counting combinations of partial solutions: we want to take the minimum over the sum of partial solution sizes.
I.e., instead of Equation~\ref{eq:countds}, we need:
\begin{equation} \label{eq:minds}
	A_i'(\vec{c}_i') = \min_{\vec{c}_l' \vee \vec{c}_r' = \vec{c}_i'} A_l'(\vec{c}_l') + A_r'(\vec{c}_r')
\end{equation}
To obtain a similar result for the minimisation versions, we will embed this into a counting structure.
\begin{lemma}[based on~\cite{vanRooijBR09}] \label{lem:minids}
	The join for the minimisation variant of {\sc Dominating Set} can be computed in $\BigO(3^kn(k + \log(n)))$ arithmetic operations.
\end{lemma}
\begin{proof}
	The proof is identical to the proof of Lemma~\ref{lem:countds}, except that we need a different fast evaluation algorithm: one that corresponds to Equation~\ref{eq:minds}.
	To this end, we expand the memoisation tables $A_i'$, $A_l'$ and $A_r'$ by having the solution size as part of the index of the table.
	That is, we let:
	\begin{equation} \label{eq:minds2}
		A_l'( \vec{c}_l', \kappa_l ) = \left\{ \begin{array}{ll}
			1 & \textrm{ if $A_l'( \vec{c}_l' ) = \kappa_l$} \\
			0 & \textrm{ otherwise} 
		\end{array} \right. 
	\end{equation} 
	and similarly for $A_r'$.
	Let $A_i'(\vec{c}_i',\kappa_i)$ be defined as follows, which can be computed using Theorem~\ref{thrm:generalmobiousjoin}:
	\begin{equation} \label{eq:mindsjoin}
		A_i'(\vec{c}_i',\kappa_i) = \sum_{\vec{c}_l' \vee \vec{c}_r' = \vec{c}_i'} \sum_{\kappa_l + \kappa_r = \kappa_i} A_l'(\vec{c}_l',\kappa_l) \, A_r'(\vec{c}_r',\kappa_r)
	\end{equation}	
	It is easy to see that $A_i'(\vec{c}_i',\kappa_i) > 0$ if and only if there exists $\vec{c}_l', \kappa_l, \vec{c}_r', \kappa_r$ such that $\vec{c}_l' \vee \vec{c}_r = \vec{c}_i$ and $A_l'(\vec{c}_l) = \kappa_l$ $A_r'(\vec{c}_r) = \kappa_r$.
	This allows us to obtain the result required by Equation~\ref{eq:minds} by setting $A_i'(\vec{c}_i)$ equal to the minimum value of $\kappa_i$ for which $A_i'(\vec{c}_i,\kappa_i) > 0$.
	
	Observe that $\kappa_i$ can range between 0 and $n$.
	Therefore, when we apply Theorem~\ref{thrm:generalmobiousjoin} with $N=\Nat_{<n+1}$, we can perform the join in $\BigO(3^kn(k + \log(n)))$ arithmetic operations.
	\qed
\end{proof}
\begin{corollary} \label{cor:mids}
	Given a graph~$G$ with a tree decomposition~$T$ of $G$ of width~$t$, {\sc Independent Dominating Set} can be solved in $\BigO(3^tn^2(t + \log(n)))$ arithmetic operations on $\BigO(t + \log(n))$-bit numbers.
\end{corollary}
\begin{proof}
	Plug Lemma~\ref{lem:minids} into Lemma~\ref{lem:bottleneck}.
	We need $\log(n)$ bit numbers for the sizes of partial solutions, while the sums in Equation~\ref{eq:mindsjoin} can require up to $\BigO(k)$-bit numbers.
	\qed
\end{proof}

We can gain linear dependence on $n$ for {\sc Dominating Set} by using a replacement property (see also~\cite{vanRooijBR09}) that holds both for {\sc Dominating Set} and {\sc Total Dominating Set}.
\begin{definition}[replacement property for partial solutions] \label{def:replacementprop}
	An optimisation problem $\mathcal{P}$ has the \emph{replacement property} if the difference in size between the smallest and the largest partial solution for non-dominated equivalence classes is at most $k$.
\end{definition}
Mostly this holds for a problem~$\mathcal{P}$ when, if given two partial solutions, one can add or subtract all vertices in $X_i$ from either one to obtain a solution that is at least as good as the other.
This is the case for {\sc Dominating Set} as adding all vertices in $X_i$ to a partial solution $D$ dominates all vertices any partial solution can dominate, thus being less restrictive than any other partial solution.

\begin{corollary} \label{cor:minds}
	Given a graph~$G$ with a tree decomposition~$T$ of $G$ of width~$t$, {\sc Dominating Set} can be solved in $\BigO(3^tt^2n)$ arithmetic operations on $\BigO(t + \log(n))$-bit numbers.
\end{corollary}
\begin{proof}
	We further modify the algorithm used in Lemma~\ref{lem:minids} (based on Lemma~\ref{lem:countds}).
	Let $\xi_l$ and $\xi_r$ be the minimum values from $A'_l(\vec{c}'_l)$ and $A'_r(\vec{c}'_r)$.
	If we restrict the ranges of $\kappa_i$, $\kappa_l$, $\kappa_r$ in Equation~\ref{eq:mindsjoin} to $[0,1,\ldots,k]$, then Theorem~\ref{thrm:generalmobiousjoin} allows us to evaluate the equation in $\BigO(2^{k'}{k'}^2)$ arithmetic operations.
	We can do so by subtracting $\xi_l$ from all values in $A'_l$ and $\xi_r$ from all values in $A'_r$ before adding the size-parameter to the index of the table.
	After the join, we can add $\xi_l + \xi_r$ to the results in $A_i'$.
	It is not hard to see that this does not influence the result of the algorithm, but now allows an $\BigO(3^kk^2)$-time join operation.
	The result then follows from plugging this result into Lemma~\ref{lem:bottleneck}.
	\qed
\end{proof}

\subsection{Count and Filter: Strong Stable Set, Perfect Code and Perfect Dominating Set} \label{sec:filtering}
To obtain fast joins for the next set of problems, we now introduce a filtering trick based on counting.
The algorithm we use here, is in essence the fast subset convolution algorithm by Bj\"orklund et al.~\cite{BjorklundHKK07}.
Our different presentation is chosen so that we can use the same trick in the sections to follow.

First notice that the three problems mentioned in this section's title have essentially the same join table (the second table in Figure~\ref{fig:jointablesds}): even though {\sc Perfect Dominating Set} uses the $|\!\geq\!0|_\sigma$-label while the others use the $|0|_\sigma$-label, the structure of the join tables is identical.
Compared to the join table for {\sc Dominating Set}, the difference in terms of Equation~\ref{eq:countds} is that we now want to compute:
\begin{equation} \label{eq:ssset}
	A_i'(\vec{c}_i') = \sum_{\vec{c}_l' + \vec{c}_r' = \vec{c}_i'} A_l'(\vec{c}_l') A_r'(\vec{c}_r')
\end{equation}
That is, where in Equation~\ref{eq:countds} we sum over three combinations to obtain a $|\!\geq\!1|_\rho$-label in $\vec{c}_i'$,
we may now only sum over two combinations ($|0|_\rho + |1|_\rho = |1|_\rho$, $|1|_\rho + |0|_\rho = |1|_\rho$).

\begin{lemma} \label{lem:ssset}
	The join for the maximisation variant of {\sc Strong Stable Set} can be computed in $\BigO(3^knk(k + \log(n)))$ arithmetic operations.
\end{lemma}
\begin{proof}
	We use the same construction as in Lemma~\ref{lem:countds} enumerating all subsets $X_\sigma \subseteq X_i$, $X' = X_i \setminus X_\sigma$, $k' = |X'|$; and for each subset~$X_\sigma$, we fix the labels of the vertices in $X_\sigma$ in $\vec{c}_l$, $\vec{c}_r$ and $\vec{c}_i$ to $|\!\geq\!0|_\sigma$.
	For the remaining subproblem, let $A_i'$, $A_l'$, $A_r'$ be the memoisation tables $A_i$, $A_l$, and $A_r$ after fixing the vertices with label $|\!\geq\!0|_\sigma$, i.e., they are indexed by state colourings~$\vec{c}_l'$ and~$\vec{c}_r'$ on $X_i'$ using only $|0|_\rho$ and $|1|_\rho$-labels.
	Next, we add the solution size to the index of these tables as in the proof of Lemma~\ref{lem:minids}: $A'_l(\vec{c}_l',\kappa_l) = 1$ if and only if $A'_l(\vec{c}_l') = \kappa_l$.
	Observe that now the join can be computed by letting $A'_i(\vec{c}_i')$ be the minimum value $\kappa_i$ for which $A'_i(\vec{c}_i',\kappa_i) > 0$, where:
	\begin{equation} \label{eq:maxsssjoin}
		A_i'(\vec{c}_i',\kappa_i) = \sum_{\vec{c}_l' + \vec{c}_r' = \vec{c}_i'} \sum_{\kappa_l + \kappa_r = \kappa_i} A_l'(\vec{c}_l',\kappa_l) \, A_r'(\vec{c}_r',\kappa_r)
	\end{equation}
	To compute the result of Equation~\ref{eq:maxsssjoin} efficiently, we add yet another parameter to the index of the tables.
	This parameter counts the number of $|1|_\rho$-labels in the state colouring~$c$.
	In other words:
	\begin{equation} \label{eq:maxsssextended}
		A_l'(\vec{c}_l',\kappa_l,\iota_l) = \left\{ \begin{array}{ll}
			A_l'(\vec{c}_l',\kappa_l) & \textrm{ if $\#_{|1|_\rho}(\vec{c}_l) = \iota_l$} \\
			0 & \textrm{ otherwise}
		\end{array} \right.
	\end{equation}
	where $\#_{|1|_\rho}(\vec{c}_l) = \iota_l$ is our notation for stating that $\vec{c}_l$ contains exactly $\iota_l$ $|1|_\rho$-labels.
	We claim that $A_i'(\vec{c}_i',\kappa_i)$ as defined in Equation~\ref{eq:maxsssjoin} equals $A_i'(\vec{c}_i',\kappa_i,\#_1(\vec{c}_i'))$, where $A_i'(\vec{c}_i',\kappa_i,\iota_i)$ is defined as:
	\begin{equation} \label{eq:maxsssjoin2}
		A_i'(\vec{c}_i',\kappa_i,\iota_i) = \sum_{\vec{c}_l' \vee \vec{c}_r' = \vec{c}_i'} \sum_{\kappa_l + \kappa_r = \kappa_i} \sum_{\iota_l + \iota_r = \iota_i} A_l'(\vec{c}_l',\kappa_l,\iota_l) \, A_r'(\vec{c}_r',\kappa_r,\iota_i)
	\end{equation}
	where $\vec{c}_l' \vee \vec{c}_r' = \vec{c}_i'$ is again defined coordinate-wise and by identifying $C_\rho$ with $\Nat_{<2}$.
	
	Notice that $\iota_l$ and $\iota_r$ track the number of $|1|_\rho$-labels used.
	Therefore, the total number of $|1|_\rho$-labels in a pair $(\vec{c}_l', \vec{c}_r')$ used as $\vec{c}_l' \vee \vec{c}_r' = \vec{c}_i'$ in a summand of the sum for $A_i'(\vec{c}_i',\kappa_i,\iota_i)$ is exactly~$\iota_i$.
	Since we need at least one $|1|_\rho$-label in $\vec{c}_l'$ or $\vec{c}_r'$ to realise each $|1|_\rho$-label in $\vec{c}_i$, we know that $A_i'(\vec{c}_i',\kappa_i,\#_{|1|_\rho}(\vec{c}_i'))$ uses $\#_{|1|_\rho}(\vec{c}_i')$ $|1|_\rho$-labels in total, and hence equals $A_i'(\vec{c}_i',\kappa_i)$ from Equation~\ref{eq:maxsssjoin}.
	
	By Theorem~\ref{thrm:generalmobiousjoin}, Equation~\ref{eq:maxsssjoin2} can be evaluated in $\BigO(2^{k'}nk'(k' + \log(n)))$ arithmetic operations.
	Summing the running time over all $2^k$ subsets $X_\sigma \subseteq X_i$ of vertices for which we fixed the label, this leads to the claimed running time in exactly the same way as in the proof of Lemma~\ref{lem:countds}.
\qed
\end{proof}

\begin{corollary} \label{cor:ssset}
	Given a graph~$G$ with a tree decomposition~$T$ of $G$ of width~$t$, the optimisation variants of {\sc Strong Stable Set}, {\sc Perfect Code} and {\sc Perfect Dominating Set} can be solved in $\BigO(3^tn^2t(t + \log(n)))$ arithmetic operations on $\BigO(t + \log(n))$-bit numbers.
\end{corollary}
\begin{proof}
	Plug Lemma~\ref{lem:ssset} into Lemma~\ref{lem:bottleneck}.
	We need $\log(n)$-bit numbers for the sizes of partial solutions, while the sums in Equation~\ref{eq:maxsssjoin2} can require up to $\BigO(k)$-bit numbers.
\qed
\end{proof}

\subsection{Fourier Transforms for Induced Bounded Degree or p-Regular Subgraph} \label{sec:regularsubgraph}
The results of the previous section can also be obtained using counting and filtering on top of Fourier transforms instead of on top of M\"obius transforms.
The resulting construction also results in fast joins for {\sc Induced Bounded Degree Subgraph} and {\sc $p$-Regular Subgraph}: these we present in this section.
The join table for {\sc 3-Regular Subgraph} is given in Figure~\ref{fig:jointablesds}: notice that it is identical to the join table for {\sc Induced Bounded Degree Subgraph} with degree bound three.

We should note that the results in this section can also be obtained using Cygan and Pilipczuk's method~\cite{CyganP10} where they encode solutions as a polynomial and use FFT-based fast polynomial multiplication.
The approach below is allows for easier combination with M\"obius transforms (as we will see in Section~\ref{sec:new}), and saves a factor of $t$ in the polynomial factors of the running time\footnote{The approach in~\cite{CyganP10} uses a factor $k^3$ (compared to our $k^2$) by performing $k^2$ FFT-based polynomial multiplications that each cost $\BigO((\ell+1)^k\log((\ell+1)^k)) = \BigO((\ell+1)^k k \log(\ell+1))$ arithmetic operations.}.

We start in the same setting as before: we fix the vertices with label $|\!\geq\!0|_\rho$ by looping over all subsets $X_\rho \subseteq X_i$ (now we fix states from $C_\rho$, previously from~$C_\sigma$) and let $X' = X_i \setminus X_\rho$ and $k'= |X'|$.
Let $A'_i$, $A'_l$ and $A'_r$ be the memoisation tables after fixing the vertices with the $|\!\geq\!0|_\rho$-label indexed by state colourings $\vec{c}_i'$, $\vec{c}_l'$, $\vec{c}_r'$ using labels from $C_\sigma = \{|0|_\sigma, |1|_\sigma, \ldots, |\ell - 1|_\sigma, |\ell|_\sigma \}$.

Define the projection function~$\pi$ on labels in $C_\sigma$ as to be $\pi(|l|_\sigma)=l$.
Also, define addition on state colourings as follows: $\vec{c}_l' + \vec{c}_r' = \vec{c}_i'$ if and only if, for all $j$, $\pi((\vec{c}_l')_j) + \pi((\vec{c}_r')_j) = \pi((\vec{c}_i')_j)$.
Computing $A_i'(\vec{c}_i')$ for the counting variant of the problem, or computing $A_i'(\vec{c}_i',\kappa_i)$ for the optimisation variant of the problem using the solution size~$\kappa$ as part of the index, now comes down to evaluating:
\begin{align} \label{eq:subgraph}
	A_i'(\vec{c}_i') = \sum_{\vec{c}_l' + \vec{c}_r' = \vec{c}_i'} A_l'(\vec{c}_l') A_r'(\vec{c}_r')
	&&
	A_i'(\vec{c}_i',\kappa_i) = \sum_{\vec{c}_l' + \vec{c}_r' = \vec{c}_i'} \sum_{\kappa_l + \kappa_r = \kappa_i} A_l'(\vec{c}_l',\kappa_l) \, A_r'(\vec{c}_r',\kappa_r)
\end{align}
Observe that this looks very similar to the statement of Lemma~\ref{lem:combinedconv}.
However, to obtain a non-cyclic convolution for $\vec{c}_l' + \vec{c}_r' = \vec{c}_i'$, a direct application of Lemma~\ref{lem:combinedconv} would require $\BigOs((\ell+1)^k2^k)$ arithmetic operations to evaluate either version of Equation~\ref{eq:subgraph}.
We will use cyclic convolution with counting and filtering to obtain the result in $\BigOs((\ell+1)^k)$ arithmetic operations, resulting in an $\BigOs((\ell+2)^k) = \BigOs(s^k)$ time join operation.

\begin{lemma} \label{lem:subgraph}
	The join for the counting variant of {\sc Induced Bounded Degree Subgraph} can be computed in $\BigO((\ell+2)^k k^2 \ell \log(\ell+1) )$ arithmetic operations.
\end{lemma}
\begin{proof}
	For a state colouring $\vec{c}_i' \in C_\sigma^{k'}$, define $\Sigma(\vec{c}_i') = \sum_{j=1}^{k'} \pi((\vec{c}_i')_j)$.
	Now, similar to the proof of Lemma~\ref{lem:ssset} where we added the number of 1-labels to the index of the table, we now add the sum of the labels $\Sigma(\vec{c}_i')$ to the index of the table.
	That is, for both $A_l'$ and $A_r'$, we set:
	\begin{equation} \label{eq:subgraphextended}
		A_l'(\vec{c}_l', \iota_l) =  \left\{ \begin{array}{ll}
			A_l'(\vec{c}_l') & \textrm{ if $\Sigma(\vec{c}_l') = \iota_l$} \\
			0 & \textrm{ otherwise}
		\end{array} \right.
	\end{equation}
	To compute the join efficiently, we use that $A_i'(\vec{c}_i')$ as defined in Equation~\ref{eq:subgraph} (left equation) equals $A_i'(\vec{c}_i',\Sigma(\vec{c}_i'))$, where $A_i'(\vec{c}_i', \iota_i)$ is the result of the following summation:
	\begin{equation} \label{eq:subgraphjoin}
		A_i'(\vec{c}_i',\iota_i) = \sum_{\vec{c}_l' + \vec{c}_r' \equiv \vec{c}_i'} \sum_{\iota_l + \iota_r = \iota_i} A_l'(\vec{c}_l',\iota_l) \, A_r'(\vec{c}_r'\iota_i)
	\end{equation}	
	where $\vec{c}_l' + \vec{c}_r' \equiv \vec{c}_i'$ is now defined as, for all $j$, $\pi((\vec{c}_l')_j) + \pi((\vec{c}_r')_j) \equiv \pi((\vec{c}_i')_j)$
	modulo $\ell+1$, and $\iota_l + \iota_r = \iota_i$ is the standard addition without modulus.
	To see that this is correct, notice that the parameters $\iota_l$ and $\iota_r$ track the sum of the labels in $\vec{c}_l'$ and $\vec{c}_r'$, and $\iota_l + \iota_r = \Sigma(\vec{c}_i')$ implies that each of the individual components of $\vec{c}_l' + \vec{c}_r'$ cannot cycle as that would result in $\iota_l + \iota_r > \Sigma(\vec{c}_i')$.
	
	By Lemma~\ref{lem:combinedconv}, Equation~\ref{eq:subgraphjoin} can be evaluated in $\BigO((\ell+1)^{k'} k'^2 \ell \log(\ell+1) )$ arithmetic operations as the $\iota$ parameters range from $0$ to $\ell k'$.
	The claimed running time follows by summing this running time over all $X_\rho \subseteq X_i$ for which we fixed the label (similar to in the proof of Lemma~\ref{lem:countds}).
\qed
\end{proof}

\begin{corollary} \label{cor:subgraph}
Given a graph~$G$ with a tree decomposition~$T$ of $G$ of width~$t$, the counting variant of {\sc Induced Bounded Degree Subgraph} can be solved in $\BigO((\ell+2)^{t+1} t^2 n \ell \log(\ell+1) )$ arithmetic operations on $\BigO(n)$-bit numbers.
\end{corollary}
\begin{proof}
Plug Lemma~\ref{lem:subgraph} into Lemma~\ref{lem:bottleneck}.
We can use $\BigO(n)$-bit numbers as the result is at most $2^n$.
Since $\ell$ is a variable (not a consant), we need $t+1$ in the exponent as $t \geq k-1$ (see Definition~\ref{def:tw}).
\qed
\end{proof}
Adapting the above lemma and corollary to the problem's optimisation variant is a simple exercise.

\subsection{M\"obius Transforms with a Different Partial Order for Total Dominating Set} \label{sec:tds}
In the previous sections, we fixed the vertices with a $|\!\geq\!0|_\sigma$-label (or $|\!\geq\!0|_\rho$-label) and used a fast transform only on vertices with a label from $C_\rho$ (or $C_\sigma$).
Here, we give an example of a fast transform that deals with vertices with different labels from both $C_\sigma$ and $C_\rho$ simultaneously.

Consider the label set $C = \{|0|_\sigma, |\!\geq\!1|_\sigma, |0|_\rho, |\!\geq\!1|_\rho \}$ associated to the {\sc Total Dominating Set} problem.
On this label set, we impose the following partial order: all labels are incomparable except that, we impose $|0|_\sigma \leq |\!\geq\!1|_\sigma$ and $|0|_\rho \leq |\!\geq\!1|_\rho$.
Notice that, using this partial order, $C^k$ does not form a lattice, e.g., in $C^3$ the join $[|0|_\sigma,|0|_\rho,|\!\geq\!1|_\sigma] \vee [|\!\geq\!1|_\rho,|\!\geq\!1|_\rho,|0|_\sigma]$ is undefined due to the first coordinate.
However, it is not hard to see that, for every $\vec{x} \in C^k$, $\{\vec{y} \in P^k | \vec{y} \leq \vec{x}\}$ forms a join-semilattice: either $x_i$ equals $|0|_\sigma$ or $|0|_\rho$ and then $y_i$ must be equal to $x_i$, or $x_i$ equals $|\!\geq\!1|_\sigma$ or $|\!\geq\!1|_\rho$ and then $y_i$ has two choices which are comparable, and thus the join is defined (actually this is a lattice as the meet is also defined).

\begin{proposition} \label{prop:tdszeta}
	Let $C$ and its partial order be defined as above.
	There are algorithms for the zeta and M\"obius transform for functions $f:C^k \rightarrow \Field_p$ that require $\BigO(4^kk)$ arithmetic operations.
\end{proposition}
\begin{proof}
	If $k=1$, the zeta and M\"obius transforms are:
	\begin{align*}
		\zeta(f)(|0|_\sigma) &= f(|0|_\sigma) & \mu(f)(|0|_\sigma) &= f(|0|_\sigma)  \\
		\zeta(f)(|\!\geq\!1|_\sigma) &= f(|0|_\sigma) + f(|\!\geq\!1|_\sigma) & \mu(|\!\geq\!1|_\sigma) &= f(|\!\geq\!1|_\sigma)- f(|0|_\sigma) \\
		\zeta(f)(|0|_\rho) &= f(|0|_\rho) & \mu(f)(|0|_\rho) &= f(|0|_\rho)  \\
		\zeta(f)(|\!\geq\!1|_\rho) &= f(|0|_\rho) + f(|\!\geq\!1|_\rho) & \mu(|\!\geq\!1|_\rho) &= f(|\!\geq\!1|_\rho)- f(|0|_\rho)
	\end{align*}
	The result follows from Lemma~\ref{lem:fastgenericzeta} as these require a constant amount of arithmetic operations. 
\qed
\end{proof}

\begin{lemma}[based on~\cite{vanRooijBR09}] \label{lem:mintds}
	The join for the minimisation variant of {\sc Total Dominating Set} can be computed in $\BigO(4^kk^2)$ arithmetic operations.
\end{lemma}
\begin{proof}
	Given $A_l$ and $A_r$, we can compute $A_i$ by evaluating the following equation that is equivalent to Equation~\ref{eq:minds} (notice that this corresponds exactly to the rightmost join table in Figure~\ref{fig:jointablesds}):
	\begin{equation} \label{eq:mintds}
		A_i(\vec{c}_i) = \min_{\vec{c}_l \vee \vec{c}_r = \vec{c}_i} A_l(\vec{c}_l) + A_r(\vec{c}_r)
	\end{equation}
	
	To obtain our result, we expand the memoisation tables $A_i$, $A_l$ and $A_r$ by having the solution size as part of the index of the table.
	That is, we let:
	\begin{equation} \label{eq:mintdsextended}
	A_l( \vec{c}_l, \kappa_l ) = \left\{ \begin{array}{ll}
	1 & \textrm{ if $A_l( \vec{c}_l ) = \kappa_l$} \\
	0 & \textrm{ otherwise} 
	\end{array} \right. 
	\end{equation}
	and similarly for $A_r$.
	Then, we compute $A_i(\vec{c}_i,\kappa_i)$ as defined below, computed using Theorem~\ref{thrm:generalmobiousjoin}:
	\begin{equation} \label{eq:mintdsjoin}
	A_i(\vec{c}_i,\kappa_i) = \sum_{\vec{c}_l \vee \vec{c}_r = \vec{c}_i} \sum_{\kappa_l + \kappa_r = \kappa_i} A_l(\vec{c}_l,\kappa_l) \, A_r(\vec{c}_r,\kappa_r)
	\end{equation}	
	Since $A_i(\vec{c}_i,\kappa_i) > 0$ if and only if there exists $\vec{c}_l, \kappa_l, \vec{c}_r, \kappa_r$ such that $\vec{c}_l \vee \vec{c}_r = \vec{c}_i$ and $A_l(\vec{c}_l) = \kappa_l$ $A_r(\vec{c}_r) = \kappa_r$, this allows us to obtain the result required by Equation~\ref{eq:mintds} by setting $A_i(\vec{c}_i)$ equal to the minimum value of $\kappa_i$ for which $A_i(\vec{c}_i,\kappa_i) > 0$.
	
	A direct application of Theorem~\ref{thrm:generalmobiousjoin} would allow us to evaluate Equation~\ref{eq:mintdsjoin} in $\BigO(4^kn^2(k+\log(n))$ arithmetic operations.
	However, if we restrict the ranges of $\kappa_i$, $\kappa_l$, $\kappa_k$ to $[0,1,\ldots,k]$, then Theorem~\ref{thrm:generalmobiousjoin} allows us to evaluate Equation~\ref{eq:mintdsjoin} in $\BigO(4^kk^2)$ arithmetic operations.
	We can do so because, just as in Corollary~\ref{cor:minds}, {\sc Total Dominating Set} satisfies the replacement property (Definition~\ref{def:replacementprop}): take the minimum value from $A_l$ and $A_r$ and subtract this from all values in $A_l$ and $A_r$ before adding the size-parameter to the index of the table.
	After the join, we can again add the sum of both minima to the results in $A_i$.
	The result now follows.
\qed
\end{proof}

\begin{corollary}
	Given a graph~$G$ with a tree decomposition~$T$ of $G$ of width~$t$, {\sc Total Dominating Set} can be solved in $\BigO(4^tt^2n)$ arithmetic operations on $\BigO(t + \log(n))$-bit numbers.
\end{corollary}
\begin{proof}
	Plug Lemma~\ref{lem:mintds} into Lemma~\ref{lem:bottleneck} and observe that we need $\log(n)$-bit numbers for the sizes of partial solutions, while the sums in Equation~\ref{eq:mintdsjoin} can require up to $\BigO(t)$-bit numbers.
\qed
\end{proof}

\section{Bringing it Together: Faster Algorithms for $[\sigma,\rho]$-Domination} \label{sec:new}
In the previous section, we have surveyed a number of approaches to realise fast joins operations and have given some concrete examples.
In this section, we will integrate several of these approaches into a new result obtaining the currently fastest algorithm for $[\sigma,\rho]$-domination in its general form.
The worst-case running time of the new algorithm is of the form $\BigO(s^t (nts)^{\BigO(1)})$, where the previously fastest algorithm by Van Rooij et al.~\cite{vanRooij11,vanRooijBR09}, has~$s$ as an exponent in the polynomial part of the running time, i.e., $\BigO(s^{t+2} (tn)^s (nts)^{\BigO(1)})$.
Here, we write $s = |C|$, where $C$ is the set of labels used in the dynamic programming algorithm for a specific $[\sigma,\rho]$-domination problem (as in Section~\ref{sec:sigmarho}).
To limit the (already heavy) notational burden, we give our result for a $[\sigma,\rho]$-domination problem with \emph{cofinite} $\sigma$ and \emph{finite} $\rho$.
It is not hard to modify the proofs for the other cases.

For state colourings $\vec{c}_l$, $\vec{c}_r$ of $X_i$ with labels from $C$, we define the operator $\vec{c}_l \oplus \vec{c}_r = \vec{c}_i$ as the coordinate-wise addition operator that keeps addition within both parts of the label set.
This operator is only defined if $\vec{c}_l$ and $\vec{c}_r$ agree on which vertices are labelled with labels from $C_\sigma$ and from $C_\rho$ and if the result again is in $C_\sigma$ or $C_\rho$.
More formally, when we write $\vec{c}_l \oplus \vec{c}_r = \vec{c}_i$ and if $(\vec{c}_l)_j = |x_l|_\rho$, $(\vec{c}_r)_j = |x_r|_\rho$, then $(\vec{c}_i)_j = |x_l + x_r|_\rho$ if $|x_l + x_r|_\rho \in C_\rho$ and otherwise it is undefined.
And, if $(\vec{c}_l)_j = |x_l|_\sigma$, $(\vec{c}_r)_j = |x_r|_\sigma$, then $(\vec{c}_i)_j = |x_l + x_r|_\sigma$ if $x_l + x_r < \ell_\sigma$ and $(\vec{c}_i)_j = |\!\geq\!\ell|_\sigma$ otherwise.
Observe that this corresponds exactly to the structure of how a join for a $[\sigma,\rho]$-domination problem with cofinite $\sigma$ and finite $\rho$ should be performed.
Besides the $\oplus$-operator, we will also use the standard $+$-operator on state colourings: $\vec{c}_l + \vec{c}_r = \vec{c}_i$.
Here, the underlying operation is the standard addition operator within each half of the label set, which is undefined if any $|\!\geq\!\ell|_\sigma$ or $|\!\geq\!\ell|_\rho$-label is involved. 
That is, if $(\vec{c}_l)_j = |x_l|_\rho$, $(\vec{c}_r)_j = |x_r|_\rho$, then $(\vec{c}_i)_j = |x_l + x_r|_\rho$, which is defined only if $|x_l + x_r|_\rho \in C_\rho$.
Addition with the $+$-operator is similar for the $\sigma$-labels.

Let again $A_l$ and $A_r$ be indexed by both the state colouring and the solution size (similar to Equations~\ref{eq:minds2} and~\ref{eq:mintdsextended}).
Performing the join for the minimisation variant of a $[\sigma,\rho]$-domination problem can now be done by extracting, for each $\vec{c}_i$, the minimum value of $\kappa_i$ for with the following expression is non-zero:
\begin{equation} \label{eq:generaljoin}
	A_i(\vec{c}_i,\kappa_i) = \sum_{\vec{c}_l \oplus \vec{c}_r = \vec{c}_i} \sum_{\kappa_l + \kappa_r = \kappa_i} A_l(\vec{c}_l,\kappa_l) A_r(\vec{c}_r,\kappa_r)
\end{equation}

To obtain a fast evaluation algorithm for Equation~\ref{eq:generaljoin}, we use both zeta/M\"obius transforms and Fourier transforms.
To do so, we impose the following partial order $\mathbbm{p}$ on the label set $C$: all labels are incomparable except that, because $\sigma$ is cofinite, we impose for all $|l|_\sigma \in C_\sigma$: $|l|_\sigma \leq |\!\geq\!\ell|_\sigma$.

Given a state colouring $\vec{c}_i$ of the vertices in $X_i$ with labels from $C = C_\sigma \cup C_\rho$, we write $\vec{c}_i = [\vec{c}^\sigma_i,\vec{c}^\rho_i] = [\vec{c}^{\geq\!\ell_\sigma}_i,\vec{c}^{<\!\ell_\sigma}_i,\vec{c}^\rho_i]$ to differentiate between the vertices with label from $C_\sigma$ and $C_\rho$, and also to further differentiate between vertices with the label $|\!\geq\!\ell|_\sigma$ and vertices with a label from $\{|0|_\sigma,|1|_\sigma,\ldots,|\ell-1|_\sigma\}$.
By splitting $\vec{c}_i$ in this way we notationally split the different coordinates.
This is just notation: we do not reorder anything in the dynamic programming table (e.g., $\vec{c}_i^\rho$ can contain the first coordinate of $\vec{c}_i$ and the last, while $\vec{c}_i^\sigma$ contains the ones in between).

Using this notation, the zeta transform of a memoisation table $A_i$ indexed by both a state colouring $\vec{c}_i = [\vec{c}^{\geq\!\ell_\sigma}_i,\vec{c}^{<\!\ell_\sigma}_i,\vec{c}^\rho_i]$ and additional indices $\vec{x}_i$ (whose purpose will become clear later) becomes:
\begin{equation}  \label{eq:mytransform0}
	\zeta(A_i)(\vec{c}_i,\vec{x}_i) 
	= \sum_{\vec{d} \leq \vec{c}_i} A_i(\vec{d}_i,\vec{x}_i)
	= \sum_{\vec{d_1} \leq \vec{c}^{\geq\!\ell_\sigma}_i} A_i([\vec{d_1},\vec{c}^{<\!\ell_\sigma}_i,\vec{c}^\rho_i],\vec{x}_i)
\end{equation}

\begin{proposition} \label{prop:mytransform}
	Given the memoisation table $A_i(\vec{c}_i,\vec{x}_i)$ indexed by state colourings $\vec{c}_i \in C^k$ over the label set $C$ ($|C|=s$) and some additional indices $\vec{x}_i$ with domain $I$, the zeta transform $\zeta(A_i)$ of $A_i$ based on the partial order $\mathbbm{p}$ can be computed in $\BigO(s^kk|I|)$ arithmetic operations.
	Also, given $\zeta(A_i)$, $A_i$ can be reconstructed in $\BigO(s^kk|I|)$ arithmetic operations.
\end{proposition}
\begin{proof}
	We will show that for $k=1$ and hence $\vec{c}_i \in C^1$, we have zeta and M\"obius transforms on $A_i(\vec{c}_i,\vec{x}_i)$ requiring $\BigO(s|I|)$ arithmetic operations.
	The result then follows from Lemma~\ref{lem:fastgenericzeta} and the fact that the transforms operate independent of the parameter~$\vec{x}_i$.
	
	By definition of the partial order $\mathbbm{p}$, the following formulas compute $\zeta(A_i)$ from $A_i$ and vice versa when $k = 1$:
	\begin{align*}
		\zeta(A_i)(\vec{c}_i,\vec{x}) &= 
		\left\{\begin{array}{ll} 
			A_i(\vec{c}_i,\vec{x}) & \textrm{ if $\vec{c}_i \neq [ |\!\geq\!\ell|_\sigma ]$} \\
			\sum_{z \in C_\sigma} A_i([z],\vec{x}) & \textrm{ if $\vec{c}_i = [ |\!\geq\!\ell|_\sigma ]$}
		\end{array} \right. \\
		A_i(\vec{c}_i,\vec{x}) &= 
		\left\{\begin{array}{ll} 
			\zeta(A_i)(\vec{c}_i, \vec{x}) & \textrm{ if $\vec{c}_i \neq [ |\!\geq\!\ell|_\sigma ]$} \\
			\zeta(A_i)(\vec{c}_i, \vec{x}) - \sum_{z \in C_\sigma \setminus \{|\!\geq\!\ell|_\sigma\}} \zeta(A_i)([z],\vec{x}) & \textrm{ if $\vec{c}_i = [ |\!\geq\!\ell|_\sigma ]$} \\
	\end{array} \right. 
	\end{align*}
	Each requires $\BigO(s|I|)$ arithmetic operations as the sums are computed only when $\vec{c}_i = [ |\!\geq\!\ell|_\sigma ]$.
\qed
\end{proof}

\begin{theorem}[bringing it all together] \label{thrm:final}
	For an optimisation variant of a $[\sigma,\rho]$-domination problem with $s=|C|$, the join can be computed in $\BigO(s^{k+1}kn(k\log(s) + \log(n)))$ arithmetic operations.
\end{theorem}
\begin{proof}
	We want to evaluate Equation~\ref{eq:generaljoin} using a fast transform.
	Consider the what happens to this equation if we apply the zeta transform based on the partial order $\mathbbm{p}$ (Equation~\ref{eq:mytransform0}) to it:
	\begin{align} \label{eqn:general1}
		\zeta(A_i)(\vec{c}_i,\kappa_i) 
		&= \sum_{\vec{d_1} \leq \vec{c}^{\geq\!\ell_\sigma}_i}
		   A_i( [\vec{d_1},\vec{c}^{<\!\ell_\sigma}_i,\vec{c}^\rho_i],\kappa_i )\\
		&= \sum_{\vec{d_1} \leq \vec{c}^{\geq\!\ell_\sigma}_i} \;\;
		   \sum_{\vec{d}_l \oplus \vec{d}_r = [\vec{d_1},\vec{c}^{<\!\ell_\sigma}_i,\vec{c}^\rho_i]} \;\;
		   \sum_{\kappa_l + \kappa_r = \kappa_i} 
		   A_l(\vec{d}_l,\kappa_l) A_r(\vec{d}_r,\kappa_r) \label{eqn:general2}
	\end{align}
	Continuing from here, we can decompose $\vec{d}_l$ and $\vec{d}_r$ coordinate-wise in the same way as we have decomposed $\vec{c}_i$ as $[\vec{c}^{\sigma\geq\!\ell}_i,\vec{c}^{\sigma<\ell}_i,\vec{c}^\rho_i]$ (to be clear: we split the coordinates of $\vec{d}_l$ and $\vec{d}_r$ based on the labels in $\vec{c}_i$, not based on the actual labels in $\vec{d}_l$ and $\vec{d}_r$).
	Let $\vec{d}_l = [\vec{d_1}_l,\vec{d_2}_l,\vec{d_3}_l]$, $\vec{d}_r = [\vec{d_1}_r,\vec{d_2}_r,\vec{d_3}_r]$.
	Now, observe that in $\vec{d}_l \oplus \vec{d}_r$, any pair $\vec{d_1}_l$ and $\vec{d_1}_r$ on the coordinates of $\vec{c}^{\geq\!\ell_\sigma}_i$ is summed over exactly once because for any pair there is exactly one $\vec{d_1}$ such that $\vec{d_1}_l \oplus \vec{d_1}_r = \vec{d_1}$.
	Also observe that because the other coordinates of $\vec{d}_l$ and $\vec{d}_r$ correspond to the vertices from the $\vec{c}^{<\!\ell_\sigma}_i$ and $\vec{c}^\rho_i$ parts of $\vec{c}_i$ their $\oplus$-sum is the standard (non-cyclic) $+$-addition on labels.
	As such, we obtain:	
	\begin{align} \label{eqn:general3}
		\zeta(A_i)(\vec{c}_i,\kappa_i)
		&= \sum_{\vec{d_1}_l, \vec{d_1}_r \leq \vec{c}^{\geq\!\ell_\sigma}_i}
		   \sum_{\vec{d_2}_l + \vec{d_2}_r = \vec{c}^{<\!\ell_\sigma}_i} 
   		   \sum_{\vec{d_3}_l + \vec{d_3}_r = \vec{c}^\rho_i}
		   \sum_{\kappa_l + \kappa_r = \kappa_i} 
		   A_l(\vec{d}_l,\kappa_l) A_r(\vec{d}_r,\kappa_r)
	\end{align}
	We note that, in the first sum, we sum over all $\vec{d_1}_l , \vec{d_1}_r \leq \vec{c}^{\geq\!\ell_\sigma}_i$ which is consistent with earlier notation, but by definition of $\mathbbm{p}$ equals all $\vec{d_1}_l, \vec{d_1}_r \in C_\sigma$.
	
	For our fast join operation, we need the sums $\vec{d_2}_l + \vec{d_2}_r$ and $\vec{d_3}_l + \vec{d_3}_r$ to be cyclic in order to use cyclic convolution.
	Therefore, we apply the same trick as in the proof of Lemma~\ref{lem:subgraph} and replace the tables  $A_l$ and $A_r$ with expanded tables that include the sums of the labels in a state colouring as an additional parameter in the index.
	Here, we do so by defining this sum of the labels of a state colouring as the sum of the number in the labels, ignoring whether they are from $C_\sigma$ or $C_\rho$ and excluding the $|\!\leq\!\ell|_\sigma$ label.
	That is, let the projection function~$\pi$ on labels to be $\pi(|l|_\sigma) =\pi(|l|_\rho)=l$.
	Then, for a state colouring $\vec{c}_i \in C^k$, define $\Sigma(\vec{c}_i) = \Sigma([\vec{c}^{\geq\!\ell_\sigma}_i, \vec{c}^{<\!\ell_\sigma}_i, \vec{c}^\rho_i]) = \sum_{j=1}^{|\vec{c}^{<\!\ell_\sigma}_i|} \pi( (\vec{c}^{\sigma<\ell}_i)_j ) + \sum_{j=1}^{|\vec{c}^\rho_i|} \pi( (\vec{c}^\rho_i)_j )$.
	For example, $\Sigma([|0|_\sigma,|0|_\rho,|\!\leq\!\ell|_\sigma] ])=0$ and $\Sigma([|2|_\sigma,|1|_\rho,|0|_\sigma] ])=3$.
	Now, we can define the expanded tables $A_l$, $A_r$ as:
	\begin{equation} \label{eq:generalextended}
		A_l(\vec{c}_l,\kappa_l, \iota_l) = \left\{ \begin{array}{ll}
		A_l(\vec{c}_l,\kappa_l) & \textrm{ if $\Sigma(\vec{c}_l) = \iota_l$} \\
		0 & \textrm{ otherwise}
		\end{array} \right.
	\end{equation}
	
	Now, we can continue from (\ref{eqn:general3}) replacing the sums with sums coordinate-wise modulo $\ell_\sigma$ and $\ell_\rho + 1$, using the additional parameter to prevent the modular-cycling to happen for the result: if cycling occurs at a coordinate, the sums do not add up any more.
	That is, we now compute $\zeta(A_i)(\vec{c}_i,\kappa_i)$ by evaluating the formula below using $\zeta(A_i)(\vec{c}_i,\kappa_i) = \zeta(A_i)(\vec{c}_i,\kappa_i,\Sigma(\vec{c}_i))$ where:
	\begin{align} \label{eqn:general4}
		\zeta(A_i)(\vec{c}_i,\kappa_i,\iota_i)
		&= \!\!\!\!\!\! \sum_{\vec{d_1}_l, \vec{d_1}_r \leq \vec{c}^{\geq\!\ell_\sigma}_i}
		   \sum_{\vec{d_2}_l + \vec{d_2}_r \equiv \vec{c}^{<\!\ell_\sigma}_i} 
		   \sum_{\vec{d_3}_l + \vec{d_3}_r \equiv \vec{c}^\rho_i}
		   \sum_{\kappa_l + \kappa_r = \kappa_i} 
		   \sum_{\iota_l + \iota_r = \iota_i} \!\!
		   A_l(\vec{d}_l,\kappa_l,\iota_l) A_r(\vec{d}_r,\kappa_r,\iota_r)
	\end{align}	
	Where $\vec{d_2}_l + \vec{d_2}_r \equiv \vec{c}^{<\!\ell_\sigma}_i$ is modulo $\ell_\sigma$and $\vec{d_3}_l + \vec{d_3}_r \equiv \vec{c}^\rho_i$ is modulo $\ell_\rho+1$ (the difference is due to the existence of the $|\!\geq\!\ell|_\sigma$-label).
	
	Next, we continue from Equation~\ref{eqn:general4} by change the order of summation, taking the outermost sum inwards, and by reordering the resulting inner terms:
	\begin{align}
		&= 	\!\! \sum_{\vec{d_2}_l + \vec{d_2}_r \equiv \vec{c}^{<\!\ell_\sigma}_i} 
			\sum_{\vec{d_3}_l + \vec{d_3}_r \equiv \vec{c}^\rho_i}
			\sum_{\kappa_l + \kappa_r = \kappa_i} 
			\sum_{\iota_l + \iota_r = \iota_i}
			\!\left( 			
				\sum_{\vec{d_1}_l \leq \vec{c}^{\geq\!\ell_\sigma}_i} 
				A_l(\vec{d}_l,\kappa_l,\iota_l) 
			\right) \!\!
			\left( 			
				\sum_{\vec{d_1}_r \leq \vec{c}^{\geq\!\ell_\sigma}_i} 
				A_r(\vec{d}_r,\kappa_r,\iota_r)
			\right)\! \\
		&= \!\!\!\!\!\!\!\!\!\!\! \sum_{\vec{d_2}_l + \vec{d_2}_r \equiv \vec{c}^{<\!\ell_\sigma}_i} 
			\sum_{\vec{d_3}_l + \vec{d_3}_r \equiv \vec{c}^\rho_i}
			\sum_{\kappa_l + \kappa_r = \kappa_i} 
			\sum_{\iota_l + \iota_r = \iota_i} \!\!\!
			\zeta(A_l)([\vec{c}^{\geq\!\ell_\sigma}_i\!\!,\vec{d_2}_l,\vec{d_3}_l],\kappa_l,\iota_l) \; \zeta(A_r)([\vec{c}^{\geq\!\ell_\sigma}_i\!\!,\vec{d_2}_r,\vec{d_3}_r],\kappa_r,\iota_r)
			\label{eqn:mytransformlast}
	\end{align}
	Where in the last step, we apply the definition of the $\zeta$-transform for $\mathbbm{p}$ (Equation~\ref{eq:mytransform0}).
	As a result, we obtain a standard convolution sum that can be evaluated using Lemma~\ref{lem:combinedconv} given that the vertices with label $|\!\geq\!\ell|_\sigma$ in $\vec{c}_i$ are fixed.
	
	To be more precise, let us partition $X_i$ into three parts $X_{\geq\!\ell_\sigma}$, $X_{<\!\ell_\sigma}$, $X_\rho$ and say that a state colouring $\vec{c}_i$ is compatible with this partition if: all vertices in $X_{\geq\!\ell_\sigma}$ have the $|\!\geq\!\ell|_\sigma$-label; all vertices in $X_{<\!\ell_\sigma}$ have a label from $C_\sigma \setminus \{|\!\geq\!\ell|_\sigma\}$; and all vertices in $X_\rho$ have a label from $C_\rho$.
	Then, given such a partition of $X_i$, Lemma~\ref{lem:combinedconv} evaluates Equation~\ref{eqn:mytransformlast} for all $\vec{c}_i$ compatible with $(X_{\geq\!\ell_\sigma},X_{<\!\ell_\sigma},X_\rho)$.
	Consequently, we can compute $\zeta(A_i)(\vec{c}_i,\kappa_i,\iota_i)$ for all values of $\vec{c}_i$, $\kappa_i$, and $\iota_i$ by enumerating all partitions of $X_i$ into $(X_{\geq\!\ell_\sigma},X_{<\!\ell_\sigma},X_\rho)$ and evaluating Equation~\ref{eqn:mytransformlast} using Lemma~\ref{lem:combinedconv} for each subset of compatible $\vec{c}_i$ values, and then taking the results together.
	
	As a result, we can evaluate Equation~\ref{eq:generaljoin} using a fast transform that takes the following steps in the following amount of operations:
	\begin{itemize}
		\item Expand the tables $A_l$ and $A_r$ taking the sums of the labels using Equation~\ref{eq:generalextended} to $A_l(\vec{c}_l,\kappa_l,\iota_l)$ and $A_r(\vec{c}_r,\kappa_r,\iota_r)$.
		This takes $\BigO(s^{k+1} kn)$ time, as $\vec{c}_l$ takes $\BigO(s^k)$ values, $\kappa_l$ takes $\BigO(n)$ values and $\iota_l$ takes $\BigO(sk)$ values.
		
		\item Compute $\zeta(A_l)$ and $\zeta(A_r)$ in $\BigO(s^{k+1}k^2n)$ arithmetic operations using Proposition~\ref{prop:mytransform}.
		
		\item Enumerate all partitions of $X_i$ into $(X_{\geq\!\ell_\sigma},X_{<\!\ell_\sigma},X_\rho)$.
		For each such partition, compute the part of the table $\zeta(A_i)(\vec{c}_i,\kappa_i,\iota_i)$ for all $\vec{c}_i$ using Equation~\ref{eqn:mytransformlast} that are compatible with this partition using Lemma~\ref{lem:combinedconv}.
		Then, combine the results to obtain $\zeta(A_i)(\vec{c}_i,\kappa_i,\iota_i)$ for all $\vec{c}_i$, $\kappa_i$ and $\iota_i$.
		
		\smallskip
		For each partition $(X_{\geq\!\ell_\sigma},X_{<\!\ell_\sigma},X_\rho)$, this takes $\BigO( (|C_\sigma|-1)^{|X_{<\!\ell_\sigma}|} |C_\rho|^{|X_\rho|} nsk(k\log(s) + \log(n)))$ arithmetic operations by Lemma~\ref{lem:combinedconv}.
		By summing over all partitions and using the multinomial theorem, we find $\BigO(s^{k+1}kn(k\log(s) + \log(n)))$ arithmetic operation for this whole step, as:
		\[ \sum_{\substack{(X_{\geq\!\ell_\sigma},X_{<\!\ell_\sigma},X_\rho)\\\textrm{partition of }X_i}} \!\!\!\! \!\!\!\! \!\! (|C_\sigma|-1)^{|X_{<\!\ell_\sigma}|} |C_\rho|^{|X_\rho|} 
		 \! = \!\!\!\! \sum_{\substack{x_1 + x_2 + x_3\\ = k}} \!\! \binom{k}{x_1, x_2, x_3} 1^{x_1} (|C_\sigma|-1)^{x_2} |C_\rho|^{x_3} \! = \! (|C_\sigma| + |C_\rho)^k \!=\! s^k
		\]

		\item Extract non-cycling values using $\zeta(A_i)(\vec{c}_i,\kappa_i) = \zeta(A_i)(\vec{c}_i,\kappa_i,\Sigma(\vec{c}_i))$.
		
		\item Compute the M\"obius transform of the result obtaining $A_i$ as $\mu(\zeta(A_i)) = A_i$.
		This is done $\BigO(s^kkn)$ arithmetic operations using Proposition~\ref{prop:mytransform}.
	\end{itemize}
	By summing over these steps we conclude that the algorithm requires $\BigO(s^{k+1}kn(k\log(s) + \log(n)))$ arithmetic operations.
\qed
\end{proof}

Above, the state colourings $\vec{c}$ are split in three components $\vec{c}_i = [\vec{c}^{\geq\!\ell_\sigma}_i,\vec{c}^{<\!\ell_\sigma}_i,\vec{c}^\rho_i]$.
If both $\sigma$ and $\rho$ are finite, we need no M\"obius transforms and it would suffice to use $\vec{c}_i = [\vec{c}^\sigma_i,\vec{c}^\rho_i]$.
If both $\sigma$ and $\rho$ are co-finite, we need M\"obius transforms for both parts of the label set and would need to use $\vec{c}_i = [\vec{c}^{\geq\!\ell_\sigma}_i,\vec{c}^{<\!\ell_\sigma}_i,\vec{c}^{\geq\!\ell_\rho}_i,\vec{c}^{<\!\ell_\rho}_i]$, and adjust the partial order $\mathbbm{p}$ in a way similar to as in Section~\ref{sec:tds}.

\begin{corollary} \label{cor:final}
	Given a graph~$G$ with a tree decomposition~$T$ of $G$ of width~$t$, the optimisation variant of a $[\sigma,\rho]$-domination problem that involves $s=|C|$ labels can be solved in $\BigO(s^{t+2}tn^2(t\log(s)+\log(n)))$ arithmetic operations on $\BigO(t\log(s) + \log(n))$-bit numbers.
\end{corollary}
\begin{proof}
	Plug Theorem~\ref{thrm:final} into Lemma~\ref{lem:bottleneck} and observe that all arithmetic operations can be done using $\BigO(t\log(s)+\log(n))$-bit numbers: the sum of all the entries in $A_l$ and $A_r$ is at most $s^k$, hence $t\log(s)$ bits, while we need the additional $\log(n)$ bits to store partial solution sizes (see also Corollary~\ref{cor:mids}).
	The $t+2$ in the exponent comes from the fact that $t \geq k-1$ (the minus one in Definition~\ref{def:tw})
\qed
\end{proof}

We conclude by summarising results for the other variants of the $[\sigma,\rho]$-domination problems.
\begin{theorem}[results for {$[\sigma,\rho]$}-domination problem variants]
Given a graph~$G$ with a tree decomposition~$T$ of $G$ of width~$t$, the different problem variants of a $[\sigma,\rho]$-domination problem involving $s=|C|$ labels can be solved with the following amount of effort:
\begin{itemize}
	\item Existence: $\BigO(s^{t+2}t^2n\log(s))$ operations on $\BigO(t\log(s))$-bit numbers.
	\item Optimisation: $\BigO(s^{t+2}tn^2(t\log(s)+\log(n)))$ operations on $\BigO(t\log(s) + \log(n))$-bit numbers.
	\item Counting: $\BigO(s^{t+2}t^2n\log(s))$ operations on $\BigO(n)$-bit numbers.
	\item Counting optimisation: $\BigO(s^{t+2}tn^2(t\log(s)+\log(n)))$ operations on $\BigO(n)$-bit numbers.
\end{itemize}
\end{theorem}
\begin{proof}
	The result for the optimisation problem follows from Corollary~\ref{cor:final}, and underlying Theorem~\ref{thrm:final}.
	
	For the counting optimisation problem, we use the same construction, only without expanding tables by solution sizes and extracting the non-zero entries: at every step of the algorithm, we let $A_i(\vec{c},\kappa)$ be the number of partial solutions of size $\kappa$ that correspond to the equivalence class identified by~$\vec{c}$.
	The join then also comes down to evaluating Equation~\ref{eq:generaljoin}, resulting in the same amount of arithmetic operations.
	For the existence and counting problems, we observe that we can remove the parameter~$\kappa$ at step of the algorithm, as solution sizes do not matter.
	Redoing the analysis of the resulting algorithm gives in the claimed amount of arithmetic operations.
	
	For both counting problems variants, we need $\BigO(n)$-bit numbers as there can be $\BigO(2^n)$ solutions to count.
	For the existence problem, we need $\BigO(t\log(s))$ as the sum of all entries in a $\BigO(s^{t+1})$ table with zero-one entries can be at most $\BigO(s^{t+1})$.
\qed
\end{proof}

\section{Conclusion}
In this paper, we have shown how M\"obius and Fourier transforms can be used to speed-up computations for dynamic programming algorithms on tree decompositions.
This led us to the currently fastest algorithm for the general case of the $[\sigma,\rho]$-domination problems on tree decompositions.
Additionally, we generalised the covering product from~\cite{BjorklundHKK07} from being defined on the subset lattice to more general partial orders (Lemma~\ref{lem:coverprod} and Theorem~\ref{thrm:generalmobiousjoin}).

The same algebraic transforms can, and have been, used for many different problems.
For example, the M\"obius-transform-based approach has been used for clique packing, partitioning and covering problems such as {\sc Partition Into Triangles} or {\sc Minimum Cover By Cliques}; see~\cite{vanRooijBR09}.
Also, the Fourier-transform-based approach (the variant from~\cite{CyganP10}) has been used for {\sc Bandwidth}~\cite{CyganP10} and {\sc Connected Vertex Cover}~\cite{vanRooijR19}.
The Fourier-transform-based approach in this paper originates from~\cite{CyganNPPRW11a}, where it was used together with counting and filtering to obtain the join table in Figure~\ref{fig:hamcyclejoin} for {\sc Longest/Hamiltonian Path/Cycle} and {\sc (Parital) Cycle Cover}.
\begin{figure}[tb]
	\begin{center}
		\scalebox{.9}{
			\begin{tabular}{c|cccc|} 
				& $0$ & $1_1$ & $2$ & $1_2$ \\ \hline 
				$0$ & $0$ & $1_1$ & $2$ & $1_2$ \\
				$1_1$ & $1_1$ & $2$ & & \\
				$2$ & $2$& & &  \\
				$1_2$ & $1_2$ & & & $2$ \\ \hline	
			\end{tabular}
			\hspace{0.2cm}
			\begin{tabular}{c|cccc|} 
				& $0$ & $1_1$ & $2$ & $1_2$ \\ \hline 
				$0$ & $0$ & $1_1$ & $2$ & $1_2$ \\
				$1_1$ & $1_1$ & $2$ & $1_2$ & $0$ \\
				$2$ & $2$ & $1_2$ & $0$ & $1_1$ \\
				$1_2$ & $1_2$ & $0$ & $1_1$ & $2$ \\ \hline	
			\end{tabular}
	}\end{center}
	\caption{The right join table is for a join for the {\sc Longest/Hamiltonian Path/Cycle} and {\sc (Parital) Cycle Cover} problems, as described in the appendix of~\cite{CyganNPPRW11a}. It is obtained by using a variant of `counting and filtering' (Section~\ref{sec:filtering}) on the left join table, for which a fast FFT-based join exists.}
	\label{fig:hamcyclejoin}
\end{figure}

We want to emphasise the more general observation that almost any join for which the join table has a certain `max' or `addition' or `modulo' structure, or a combination of those, can be done fast using the tools from this paper.
For example, the Fourier transform and corresponding cyclic convolution theorem can be used to obtain algorithms for problems where there is some modulo relation in the definition of the problem's solution set~$D$, e.g., an odd number of neighbours in~$D$.

The approaches in this paper have wider use beyond tree decompositions.
For example, they can be applied to branch decomposition instead of tree decompositions, obtaining faster algorithms there as well, with faster exact $\BigO(c^{\sqrt{n}})$-time algorithms on planar graphs as direct corollaries~\cite{BodlaenderLRV10,PinoBR16}.
 
\paragraph{Open Problems.}
What we see is that, in order to use fast algebraic transforms, we embed the problem into algebraic structures that we further parameterise by the solution size ($\kappa$ in the algorithms in this paper).
However, without the replacement property (Definition~\ref{def:replacementprop}), this leads to algorithms with super-linear dependence on~$n$, while algorithms that are exponentially-slower in~$t$ but linear in~$n$ exists.
Can we somehow remove this super-linear dependence on $n$?

Moreover, when we consider weighted versions of the problems, the weights will appear in the run times of the exponentially-optimal algorithms.
For the exponentially-slower algorithms (e.g., those by Alber et al.~\cite{AlberBFKN02}) weights play no role in the worst-case running times.
Can we somehow remove the dependence on the weights and obtain $\BigOs(s^t)$-time algorithms for weighted problems?

\end{document}